\documentclass{article}

\usepackage[utf8]{inputenc}
\usepackage[T1]{fontenc}

\usepackage[margin=1.6in]{geometry}

\usepackage{amsthm, amsmath, amssymb}
\usepackage{eucal}

\usepackage{needspace}
\usepackage{float}
\usepackage{thmtools, thm-restate}
\usepackage{enumerate}
\usepackage{hyperref}
\usepackage{graphicx}
\usepackage{authblk} 
\usepackage{titlesec} 

\usepackage{xcolor}
\definecolor{citation}{HTML}{e55039}

\hypersetup{
  pdftitle = {},
  pdfauthor = {TODO},
  colorlinks = true,
  linkcolor = black!10!red, 
  citecolor = black!10!red
}

\newcommand{\pdot}{${\cdot}$} 
 
\newcommand{\icsenum}{{\textsc{ICS}\pdot\textsc{Enum}}}
\newcommand{\acsenum}{{\textsc{ACS}\pdot\allowbreak{}\textsc{Enum}}}

\newcommand{\acsaenum}{{\textsc{ACS\texttt{+}A}\pdot\allowbreak{}\textsc{Enum}}}
\newcommand{\cgeagen}{{\textsc{CGE\texttt{+}A}\pdot\allowbreak{}\textsc{Gen}}}
\newcommand{\cgeadec}{{\textsc{CGE\texttt{+}A}\pdot\textsc{Dec}}}

\newcommand{\misenum}{{\textsc{MIS}\pdot\textsc{Enum}}}
\newcommand{\mhsenum}{{\textsc{MHS}\pdot\textsc{Enum}}}

\newcommand{\mibgen}{{\textsc{MIB}\pdot\textsc{Gen}}}
\newcommand{\critbase}{{\textsc{CB}\pdot\textsc{Gen}}}

\newcommand{\icsext}{{\textsc{ICS}\pdot\textsc{Ext}}}

\newcommand{\algoa}{\mathsf{A}}
\newcommand{\algob}{\mathsf{B}}

\newcommand{\N}{\mathcal{N}}
\newcommand{\NP}{{\sf NP}}
\newcommand{\XP}{{\sf XP}}
\newcommand{\FPT}{{\sf FPT}}

\newcommand{\F}{\mathcal{F}}

\renewcommand{\H}{\mathcal{H}}

\newcommand{\E}{\mathcal{E}}

\newcommand{\C}{\mathcal{C}}  
\renewcommand{\S}{\mathcal{S}}  

\newcommand{\GC}{\mathrm{{GC}}}

\newcommand{\poly}{\mathrm{poly}}
\newcommand{\crit}{\mathrm{crit}}
\newcommand{\critgen}{\mathrm{critgen}}
\newcommand{\critanc}{\mathrm{critanc}}
\newcommand{\mingen}{\mathrm{mingen}}
\newcommand{\anc}{\mathrm{anc}}
\newcommand{\card}[1]{\vert #1 \vert}  

\newcommand{\pow}[1]{\mathcal{P}(#1)}

\newcommand{\pdeg}{\mathrm{pdeg}}
\newcommand{\cdeg}{\mathrm{cdeg}}
\newcommand{\MIS}{\mathrm{MIS}}
\newcommand{\MHS}{\mathrm{MHS}}
\newcommand{\irr}{\mathrm{irr}}
\newcommand{\ex}{\mathrm{ex}}

\usepackage{tabularx}
\usepackage{environ}
\usepackage{makecell}
\makeatletter
\newcommand{\problemtitle}[1]{\gdef\@problemtitle{\normalsize{}#1}}%
\newcommand{\problemparameter}[1]{\gdef\@problemparameter{\normalsize{}#1}}
\newcommand{\probleminput}[1]{\gdef\@probleminput{\normalsize{}#1}}
\newcommand{\problemquestion}[1]{\gdef\@problemquestion{\normalsize{}#1}}
\NewEnviron{problemgen}{
\problemtitle{}\problemparameter{}\probleminput{}\problemquestion{}
\BODY
\par\addvspace{.5\baselineskip}
\noindent
\begin{tabularx}{\textwidth}{@{\hspace{\parindent}} l X c}
    \multicolumn{2}{@{\hspace{\parindent}}l}{\@problemtitle} \\
    \textbf{Input:} & \@probleminput \\
    \textbf{Output:} & \@problemquestion
  \end{tabularx}
  \par\addvspace{.5\baselineskip}
}

\NewEnviron{problemdec}{
\problemtitle{}\probleminput{}\problemquestion{}
  \BODY
  \par\addvspace{.5\baselineskip}
  \noindent
  \begin{tabularx}{\textwidth}{@{\hspace{\parindent}} l X c}
    \multicolumn{2}{@{\hspace{\parindent}}l}{\@problemtitle} \\
    \textbf{Input:} & \@probleminput \\
    \textbf{Question:} & \@problemquestion
  \end{tabularx}
  \par\addvspace{.5\baselineskip}
}

\usepackage{tikz}
\usetikzlibrary{shapes}

\usepackage{thm-restate}

\declaretheorem[parent=section, name=Theorem, style=plain]{theorem}
\declaretheorem[sibling=theorem, name=Lemma, style=plain]{lemma}
\declaretheorem[sibling=theorem, name=Corollary, style=plain]{corollary}
\declaretheorem[sibling=theorem, name=Proposition, style=plain]{proposition}
\declaretheorem[sibling=theorem, name=Question, style=plain]{question}

\declaretheorem[sibling=theorem, name=Example, style=definition]{example}
\declaretheorem[sibling=theorem, name=Observation, style=definition]{observation}

\declaretheorem[sibling=theorem, name=Remark, style=remark]{remark}

\newcommand{\myparagraph}[1]{\paragraph{#1}}

\renewenvironment{abstract}
{\small\vspace{-1em}
\begin{center}
\bfseries\abstractname\vspace{-.5em}\vspace{0pt}
\end{center}
\list{}{
\setlength{\leftmargin}{0.6in}%
\setlength{\rightmargin}{\leftmargin}}%
\item\relax}
{\endlist}

\usepackage{lineno}

\newif\iflongversion
\longversiontrue

\title{%
Translating between the representations\\ of an acyclic convex geometry of bounded degree\thanks{An extended abstract of this work has been accepted for publication at the 36th International Symposium on Algorithms and Computation (ISAAC 2025). The authors have been supported by the ANR project PARADUAL (ANR-24-CE48-0610-01).}}

\author[1]{Oscar Defrain}
\author[1,2]{Arthur Ohana}
\author{Simon Vilmin}

\affil[1]{Aix-Marseille Université, CNRS, LIS, Marseille, France.}
\affil[2]{Université Grenoble Alpes, France}

\begin{document}

\maketitle

\vspace{-.3cm}

\begin{abstract}

We consider the problem of translating between irreducible closed sets and implicational bases in closure systems.
To date, the complexity status of this problem is widely open, and it is further known to generalize the notorious hypergraph dualization problem, even in the case of acyclic convex geometries, i.e., closure systems admitting an acyclic implicational base.
This paper studies this case with a focus on the degree, which corresponds to the maximal number of implications in which an element occurs.
We show that the problem is tractable for bounded values of this parameter, even when relaxed to the notions of premise- and conclusion-degree.
Our algorithms rely on structural properties of acyclic convex geometries and involve various techniques from algorithmic enumeration such as solution graph traversal, saturation techniques, and a sequential approach leveraging from acyclicity.
They are shown to perform in incremental-polynomial time.
Finally, we complete these results by showing that our running times cannot be improved to polynomial delay using the standard framework of flashlight search.

\vskip5pt\noindent{}{\bf Keywords:} algorithmic enumeration, closure systems, acyclic convex geometries, implicational bases, irreducible closed sets, incremental-poly\-nomial, saturation algorithms, solution graph traversal, flashlight search.
\end{abstract}

\section{Introduction}

Finite closure systems 
appear in disguise in several areas of mathematics and computer science, including database theory~\cite[Chapter 13]{mannila1992design}, Horn logic~\cite[Chapter 6]{crama2011boolean}, or lattice theory~\cite[Chapter 2]{davey2002introduction}, to mention but a few.

A closure system is a family of subsets of some finite ground set being closed by intersection and containing the ground set.
The study of the various representations of closure systems and the algorithmic task of translating between them have received lots of attention in these respective fields, and have been the topic of the Dagstuhl Seminar 14201~\cite{dagstuhl2014}.  
The translation task we study in this paper is most importantly known for being a considerable generalization of the notorious and widely open hypergraph dualization problem, while also not being known to be intractable, contrary to analogous other natural generalizations such as lattice dualization~\cite{kavvadias2000generating, babin2017dualization,defrain2020dualization}.
We redirect the reader to the recent works~\cite{distel2011complexity,babin2013computing,beaudou2017algorithms,defrain2021translating,nourine2023hierarchical} as well as to the surveys~\cite{wild2017joy,bertet2018lattices, rudolph2017succinctness} for more information on the topic.

Among the several representations of closure systems, two of them occupy a central role, namely, \emph{irreducible closed sets} and \emph{implicational bases}, whose formal definitions are postponed to Section~\ref{sec:preliminaries}.
Informally, irreducible closed sets are particular sets of the closure system, while implicational bases are sets of implications of the form $A \to B$, where $A$, the premise, and $B$, the conclusion, are subsets of the ground set.
A closure system has a unique collection of irreducible closed sets, but it can be represented by several equivalent implicational bases.
As noted in~\cite{khardon1995translating}, these two representations are generally incomparable, in size and usefulness. 
Notably, there exist instances where the minimum cardinality of an implicational base is exponential in the number of irreducible closed sets, and vice versa. 
Moreover, the computational complexity of some algorithmic tasks such as reasoning, abduction, or recognizing lattice properties can differ significantly depending on the representation at hand~\cite{kautz1993reasoning,babin2013computing,mannila1992design}. 
These observations motivate the problem of translating between the two representations, which has two sides: enumerate the irreducible closed sets from an implicational base, and produce a compact implicational base from the collection of irreducible closed sets.
These problems are respectively denoted by \icsenum{} and \mibgen{} throughout the paper and formally defined in the next section.

Because of the possible exponential gap between the sizes of the two representations, input-sensitive complexity is not a meaningful efficiency criterion when analyzing the performance of an algorithm solving \icsenum{} or \mibgen{}.
Instead, the framework of algorithmic enumeration in which an algorithm has to list without repetitions a set of solutions, is usually adopted.
The complexity of an algorithm is then analyzed from the perspective of \emph{output-sensitive complexity}, which estimates the running time of the algorithm using both input and output sizes.
In that context, an algorithm is considered tractable when it runs in \emph{output-polynomial} time, that is, if its execution time is bounded by a polynomial in the combined sizes of the input and the output.
This notion can be further constrained to guarantee some regularity in the enumeration.
Namely, we say that an algorithm runs in \emph{incremental-polynomial time} if it outputs the $i^\text{th}$ solution in a time which is bounded by a polynomial in the input size plus $i$.
It is said to run with \emph{polynomial delay} if the time spent before the first output, after the last output, and between consecutive outputs is bounded by a polynomial in the input size only.
We refer the reader to~\cite{johnson1988generating,strozecki2019survey} for more details on enumeration complexity. 

In this work, we investigate \emph{acyclic convex geometries}, a particular class of closure systems.
They are the convex geometries---closure systems with the anti-exchange property---that admit an implicational base with acyclic implication graph.
Acyclic convex geometries have been extensively studied~\cite{hammer1995quasi,wild1994theory,boros2009subclass,santocanale2014lattices,defrain2021translating}.
Even in this class though, the two problems \icsenum{} and \mibgen{} have been shown in \cite{defrain2021translating} to be harder than distributive lattices dualization, a generalization of hypergraph dualization.
On the other hand, the authors in \cite{defrain2021translating} show that if, in addition to acyclicity, the implication graph admits a rank function, then the problem can be solved using hypergraph dualization.
Using this result, they derive a tractable algorithm for so-called ranked convex geometries admitting an implicational base of bounded premise size.
It should be noted that, even under these restrictions, the problem generalizes the one of enumerating all maximal independent sets of a graph (in fact of hypergraphs of bounded edge size) a central and notorious problem in enumeration~\cite{tsukiyama1977new,johnson1988generating}.
These results have then been extended to a more general setting in~\cite{nourine2023hierarchical}.
Let us finally note that, outside of acyclic convex geometries, the translation task is known to be tractable in modular or $k$-meet-semidistributive lattices~\cite{wild2000optimal, beaudou2017algorithms}, or in closure spaces of bounded Caratheodory number~\cite{wild2017joy}.

In this paper, we continue the line of research of characterizing tractable cases of acyclic convex geometries for \icsenum{} and \mibgen{}, typically, by identifying for which parameters $k$ we can obtain tractable algorithms for these problems for fixed values of $k$.
We note that in this research, 
finding an algorithm running in quasi-polynomial time should already be considered a success,
since no sub-exponential time algorithms are known to date.
We focus on the degree, which is defined as the maximum number of implications in which an element participates, and its two obvious relaxations which are the premise- and conclusion-degree; these parameters are defined in Section~\ref{sec:preliminaries}.
The principal message of this work is that bounding one of these two relaxed notions of degree suffices to provide not only quasi-polynomial but polynomial-time bounds for the translation.
More specifically, as the first of our two main theorems, we obtain the following, whose quantitative formulation is postponed to Sections~\ref{sec:conclusion-degree} and \ref{sec:premise-degree}.

\begin{restatable}{theorem}{restatethmicsenum}
\label{thm:icsenum}
    There is an incremental-polynomial time algorithm enumerating the irreducible closed sets of an acyclic convex geometry given by an acyclic implicational base of bounded premise- or conclusion-degree.
\end{restatable}

Our algorithms for acyclic convex geometries of bounded degree rely on several steps which combine structure and enumeration techniques.
We give a rough description which highlights the techniques used, and refer to the appropriate sections for more details.

For generating irreducible closed sets from an acyclic implicational base, we use the fact that in a convex geometry, each irreducible closed set is \emph{attached} to a unique element of the ground set; this is detailed in Section~\ref{sec:preliminaries}.
Then, leveraging from acyclicity, our algorithm relies on a sequential procedure which consists of enumerating, incrementally and in a top-down fashion according to a topological ordering $x_1,\dots,x_n$ of the elements of the ground set, all the irreducible closed sets attached to $x_i$.
This step can be seen as enumerating the irreducible closed sets attached to $x_i$ with additional information as input, namely, the irreducible closed sets of ancestor elements.
Finally, it remains to solve this latter task, which we can do either by brute forcing the minimal transversals of a carefully designed hypergraph in case of bounded conclusion-degree, or performing a solution graph traversal in case of bounded premise-degree.
In the end, the resulting algorithms run in incremental-polynomial time, where the dependence in the number of solutions arises from the sequential top-down procedure.

We then turn to the opposite translation task of generating a compact implicational base given the irreducible closed sets of an acyclic convex geometry. 
To do so, we rely on the concept of critical generators~\cite{korte1983shelling, wild1994theory, hammer1995quasi}.
These are known to form an acyclic implicational base of minimum cardinality (in fact, minimum among those with conclusions of size one, or minimum once aggregated, a point that is discussed in Section~\ref{sec:acyclic}) in the class of acyclic geometries.
This implicational base, referred to as the critical base, thus constitutes a candidate of choice to be generated from irreducible closed sets.
We obtain the following as our second main theorem.

\begin{restatable}{theorem}{restatethmmibgen}
\label{thm:critgen}
    There is a polynomial-time algorithm constructing the critical base of an acyclic convex geometry of bounded premise- or conclusion-degree given by its irreducible closed sets.
\end{restatable}

This second algorithm also relies on the top-down approach 
except that, this time, we generate all the critical generators of $x_i$ at each step.
However, since the implicational base is not given, a particular care is devoted to arguing of the choice of the ordering $x_1,\dots,x_n$ and the fact that it can be computed from the irreducible closed sets only.
Then, the enumeration of critical generators is conducted
using a saturation algorithm relying on the existence of an algorithm for the sequential enumeration of irreducible closed sets we obtained in order to prove Theorem~\ref{thm:icsenum}.
We note that in this direction, our algorithm is much more efficient thanks to the fact that the output has polynomial cardinality.

In order to prove Theorems~\ref{thm:icsenum} and \ref{thm:critgen}, we build upon a number of existing properties of acyclic convex geometries coming from different fields.
This results in a substantial presentation of acyclic convex geometries which aims to be self-contained and accessible.
Besides, we also extend this presentation with a study of the degree in this class.

The remainder of the paper is organized as follows. 
Section~\ref{sec:preliminaries} introduces the concepts and notations used throughout the paper.
Section~\ref{sec:acyclic} continues the state of the art via a focus on acyclic convex geometries, and include additional observations that may be of independent interest.
Section~\ref{sec:top-down} presents the top-down approach that we use as the key ingredient in our algorithms.
Sections~\ref{sec:conclusion-degree} and \ref{sec:premise-degree} are devoted to proving Theorem~\ref{thm:icsenum} in the respective cases of bounded conclusion- and premise-degree.
Theorem~\ref{thm:critgen} is proved in Section~\ref{sec:mibgen}.
Finally, we discuss possible generalizations, limitations, and open directions in Section~\ref{sec:discussion}.

\section{Preliminaries} \label{sec:preliminaries}

All the objects we consider in this paper are finite.
If $X$ is a set, $\pow{X}$ is its powerset.
The size of $X$ is denoted $\card{X}$.
A \emph{set system}, or \emph{hypergraph}, over a ground set $X$ is a pair $(X, \F)$ where $\F$ is a collection of subsets of $X$, that is $\F \subseteq \pow{X}$.
Its \emph{size} is $|X|+|\F|$; note that the binary length of any standard encoding of a hypergraph is polynomially related to its size.
Thus, for simplicity, we shall consider this measure for families as input size in the different problems we will consider.
If $X$ is clear from the context, we may identify $\F$ with $(X, \F)$.
In examples, if there is no ambiguity, we may write a set as the concatenation of its elements, for instance $123$ instead of $\{1, 2, 3\}$. 
Similarly, to avoid cumbersome notations, we sometimes withdraw brackets when dealing with singletons, e.g., we may write $\phi(x)$ instead of $\phi(\{x\})$.

\myparagraph{Independent and hitting sets, hypergraph parameters.}
Let $\H = (X, \mathcal{E})$ be a hypergraph.
We call \emph{vertices} of $\H$ the elements in $X$, and \emph{edges} the sets in $\mathcal{E}$.
We say that $\H$ is \emph{Sperner} if any two of its edges are pairwise incomparable by inclusion.
An \emph{independent set} of $\H$ is a set $I \subseteq X$ that contains no edges of $\H$, i.e., $E \nsubseteq I$ for every $E$ in $\mathcal{E}$.
The family of inclusion-wise maximal independent sets of $\H$ is denoted $\MIS(\H)$.
Formally:
\[ 
\MIS(\H) := \max_\subseteq \{I : \text{$I \subseteq X$, $E \nsubseteq I$ for every $E$ in $\mathcal{E}$} \}.
\]
Dually, a \emph{hitting set}, or \emph{transversal}, of $\H$ is a subset $T$ of $X$ that intersects each edge of $\H$.
The collection of inclusion-wise minimal hitting sets of $\H$ is called $\MHS(\H)$:
\[ 
\MHS(\H) := \min_\subseteq \{T : \text{$T \subseteq X$, $T \cap E \neq \emptyset$ for every $E$ in $\mathcal{E}$}\}.
\]
The hypergraph $\MHS(\H)$ is usually called the \emph{dual hypergraph} of $\H$.
Observe that $I \subseteq X$ is independent if and only if $T := X \setminus I$ is a hitting set.
Consequently, $\MHS(\H) = \{X \setminus I : I \in \MIS(\H)\}$. 
Moreover, if $\H$ is not Sperner, then its minimal hitting sets (resp.\ maximal independent sets) are precisely those of the hypergraph formed by the inclusion-wise minimal edges of $\H$.
The two equivalent generation problems associated to maximal independent sets and minimal hitting sets read as follows:
\begin{problemgen}
\problemtitle{Maximal Independent Sets Enumeration (\misenum)}
\probleminput{A (Sperner) hypergraph $\H = (X, \mathcal{E})$.}
\problemquestion{The family $\MIS(\H)$.}
\end{problemgen}
\begin{problemgen}
\problemtitle{Minimal Hitting Sets Enumeration (\mhsenum)}
\probleminput{A (Sperner) hypergraph $\H = (X, \mathcal{E})$.}
\problemquestion{The family $\MHS(\H)$.}    
\end{problemgen}
To date, the best algorithm for these tasks runs in output-quasi-polynomial time and is due to Fredman and Khachiyan~\cite{fredman1996complexity}.
However, a number of output-polynomial time algorithms have been proposed for particular cases, especially when bounding structural parameters~\cite{eiter2008computational}.
Among these parameters, we will be interested in the \emph{dimension}, which is the size of the largest edge, or the \emph{dual dimension}, being the dimension of the dual hypergraph.

\myparagraph{Closure systems, closure spaces.} A \emph{closure system} is a set system $(X, \C)$ that satisfies $X \in \C$, and $C \cap C' \in \C$ whenever $C, C'\in \C$.
The sets in $\C$ are called \emph{closed (sets)}.
A \emph{closure space} is a pair $(X, \phi)$ where $\phi \colon \pow{X} \to \pow{X}$ is a mapping that satisfies for every $A, B \subseteq X$: (1) $A \subseteq \phi(A)$, (2) $\phi(A) \subseteq \phi(B)$ if $A \subseteq B$, and (3) $\phi(\phi(A)) = \phi(A)$.
The map $\phi$ is a \emph{closure operator}.
Closure spaces and closure systems are in one-to-one correspondence.
More precisely, if $(X, \phi)$ is a closure space, the family $\C := \{C : C \subseteq X,\ \phi(C) = C\}$ defines a closure system.
Dually, if $(X, \C)$ is a closure system, the map $\phi$ defined by $\phi(A) := \min_{\subseteq} \{C : C \in \C,\ A \subseteq C\} = \bigcap \{C \colon C \in \C,\ A \subseteq C\}$ is a closure operator whose associated closure system is precisely $(X, \C)$.
In this paper, if $(X, \phi)$ is a closure space, we call $(X, \C)$ its \emph{corresponding} closure system, and vice-versa. 
Moreover, we always consider without loss of generality that $\emptyset \in \C$, a standard assumption.

Let $(X, \C)$ be a closure system.
The poset $(\C, \subseteq)$ of closed sets ordered by inclusion is a \emph{(closure) lattice} where $C \land C' = C \cap C'$ and $C \lor C' = \phi(C \cup C')$.
A \emph{predecessor} of a closed set $C$ is an inclusion-wise maximal closed set $F$ properly contained in $C$, i.e., $F \subset C$ and for every $C' \in \C$, $F \subseteq C' \subseteq C$ implies $C' = F$ or $C' = C$.
\emph{Successors} are defined dually.
Let $M \in C$, $M \neq X$. 
The closed set $M$ is \emph{irreducible} if it is not the intersection of other distinct closed sets, that is, $M = C \cap C'$ for some $C, C' \in \C$ only if $C = M$ or $C' = M$.
Irreducible closed sets are in fact the members of $(\C, \subseteq)$ with a unique successor in the poset.
We denote by $\irr(\C)$ the family of irreducible closed sets of $(X, \C)$.
For every $C \in \C$, it is known that $C = \bigcap \{M : M \in \irr(\C),\ C \subseteq M\}$.
Moreover, $\irr(\C)$ is the unique minimal subset of $\C$ from which $\C$ can be reconstructed using intersections.
Let $x \in X$.
We denote by $\irr(x)$ the inclusion-wise maximal closed sets not containing $x$, i.e.,
\[ \irr(x) := \max_{\subseteq} \{M : M \in \C,\, x \notin M\}. \]
Observe that $\irr(x) \subseteq \irr(\C)$ for every $x \in X$.
An irreducible closed set $M$ in $\irr(x)$ is said to be \emph{attached} to $x$.
In fact, it is well known that $\irr(\C) = \bigcup_{x \in X} \irr(x)$.
A \emph{(non-trivial) minimal generator} of $x$ is an inclusion-wise minimal subset $A$ of $X \setminus \{x\}$ satisfying $x \in \phi(A)$.
The family of minimal generators of $x$ is denoted $\mathrm{mingen}(x)$.
Again, it is well-known~\cite{mannila1992design} that $\irr(x)$ and $\mathrm{mingen}(x)$, seen as hypergraphs over $X$, are dual to one another, that is:
\begin{align*}
\irr(x) & = \MIS(\mathrm{mingen}(x) \cup \{\{x\}\}),\\ 
\mathrm{mingen}(x) & = \MHS(\{X \setminus (M \cup \{x\}) : M \in \irr(x)\}).
\end{align*} 
An example of a closure system and its irreducible closed sets is depicted in Figure \ref{fig:prelim-example}.
Let $A \subseteq X$.
An \emph{extreme point} of $A$ is an element $x \in A$ such that $x \notin \phi(A \setminus \{x\})$.
We denote by $\ex(A)$ the extreme points of $A$.
In particular, if $C \in \C$ and $x \in \ex(C)$, then $C \setminus \{x\} \in \C$.
Note that for every $C, C' \in \C$ such that $C \subseteq C'$, the extreme points of $C'$ that belong to $C$ must also be extreme points of $C$, i.e., $\ex(C') \cap C \subseteq \ex(C)$.
Indeed, $x \in \ex(C')$ entails $C' \setminus \{x\} \in \C$ and hence $C \cap C' \setminus \{x\} = C \setminus \{x\} \in \C$, that is, $x \in \ex(C)$.
Dually, if $x \notin \ex(C)$ for some $C \in \C$, then $x \notin \ex(C')$ for every $C' \in \C$ such that $C \subseteq C'$.
If $C \in \C$, a \emph{spanning set} of $C$ is a subset $A$ of $X$ satisfying $\phi(A) = C$.
\begin{figure}[ht!]
    \centering
    \includegraphics[page=1, scale=0.8]{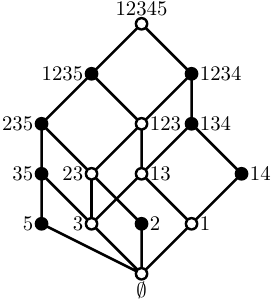}
    \caption{A closure system over $X = \{1, 2, 3, 4, 5\}$ represented via the Hasse diagram of its closure lattice. Black dots are irreducible closed sets; they are the ones with a unique successor. We have, for instance, $\phi(25) = 235$ and $\phi(4) = 14$, $\irr(2) = \{35, 134\}$ and $\mathrm{mingen}(3) = \{25, 15, 24, 45\}$. A spanning set of $1235$ is $135$. The extreme points of $135$ are $1$ and $5$.}
    \label{fig:prelim-example}
\end{figure}

\myparagraph{Convex geometries.} 
A closure system $(X, \C)$ is a \emph{convex geometry} if $\emptyset \in \C$ and for every closed set $C$ distinct from $X$ there exists $x \notin C$ such that $C \cup \{x\}$ is closed.
The closure system depicted in Figure \ref{fig:prelim-example} is a convex geometry.
Equivalently, $(X, \C)$ is a convex geometry if the corresponding closure operator $\phi$ has the \emph{anti-exchange} property: for every closed set $C$ and each distinct $x, y \notin C$, $x \in \phi(C \cup \{y\})$ entails $y \notin \phi(C \cup \{x\})$.
Convex geometries have been rediscovered in several fields and, as such, enjoy several characterizations~\cite{monjardet1985use}.
We summarize some of them in the subsequent statement:
\begin{proposition}[see e.g.,~{\cite[Theorems 2.1, 2.4]{edelman1985theory}}] \label{prop:cg-meet-partition}
Let $(X, \C)$ be a closure system such that $\emptyset \in \C$.
The following are equivalent:
\begin{enumerate}
    \item $(X, \C)$ is a convex geometry
    \item the associated closure operator $\phi$ has the anti-exchange property
    \item each irreducible closed set is attached to a unique element, i.e., $\{\irr(x) : x \in X\}$ is a partition of $\irr(\C)$
    \item each closed set $C$ has a unique minimal spanning set, which is $\ex(C)$
\end{enumerate}
\end{proposition}
Let $T\subseteq X$.
We call \emph{irreducible selection} for $T$ a family of sets obtained by choosing, for each $t\in T$, one $M\in \irr(t)$.
Note that since we are dealing with convex geometries, and by Proposition~\ref{prop:cg-meet-partition}, the closed sets $M$ are all distinct and the resulting family has the same cardinal as $T$.
In the following, we note $\S(T)$ the family of all irreducible selections for a given set $T$. 
We will use irreducible selections as a way to forbid elements, as stated in the following observation. 

\begin{observation}\label{obs:selection-intersection}
    For any $T\subseteq X$ and any selection $\S\in \S(T)$,
    the intersection $\bigcap \S$ of all sets in $\S$ is a closed set disjoint from $T$. 
\end{observation}

We conclude the preliminaries with critical generators.
These have first been defined in terms of circuits of antimatroids (complements of convex geometries) within the context of combinatorial optimization~\cite{korte1983shelling, korte2012greedoids, dietrich1987circuit} but they also have been studied from the perspective of closure systems~\cite{wild1994theory, nakamura2013prime, wild2017joy, yoshikawa2017representation, defrain2021translating}.
They will play a crucial role in this paper. 
Recall that $(X, \C)$ is a convex geometry.
Let $x \in X$ and let $A \in \mathrm{mingen}(x)$.
The minimal generator $A$ is a \emph{critical generator} of $x$ if $\phi(A) \setminus \{a, x\} \in \C$ for every $a \in A$. 
For $x \in X$, we call $\critgen(x)$ the set of critical generators of $x$. 
Using the properties of extreme points previously mentioned, we can rephrase the definition of critical generators in a way that will be more convenient to use in our approach:
\begin{proposition} \label{prop:carac-crit}
The set $A$ is a critical generator of $x$ if and only if $A = \ex(\phi(A))$ and $\phi(A) \in \min_\subseteq \{C : C \in \C,\ x \in C,\ x \notin \ex(C)\}$.
\end{proposition}
\begin{proof}
We start with the if part.
By assumption, for every $a \in A$, $x \in \ex(\phi(A) \setminus \{a\})$, i.e., $\phi(A) \setminus \{a, x\} \in \C$ for each $a \in A$.
Hence, we only need to show that $A \in \mingen(x)$.
Observe first that $x \notin A$ by definition of $A$.
Now, let $a \in A$.
We have $\phi(A \setminus \{a\}) \subseteq \phi(A) \setminus \{a\} \in \C$.
Moreover, since $A$ is the unique minimal spanning set of $\phi(A)$ by assumption, any subset of $A$ must be the unique minimal spanning set of its closure, so in particular $\ex(\phi(A \setminus \{a\})) = A \setminus \{a\}$.
Thus, $x \in \phi(A \setminus \{a\})$ would imply that $x \notin \ex(\phi(A \setminus \{a\}))$, as $x \notin A$, 
which in turn yields $x \notin \ex(\phi(A) \setminus \{a\})$, a contradiction to the minimality of $\phi(A)$.
We deduce $x \notin \phi(A \setminus \{a\})$ and hence that $A \in \mingen(x)$.

As for the only if part, $A$ is a minimal generator of $x$, hence it must be the unique minimal spanning set of $\phi(A)$.
Since $(X, \C)$ is a convex geometry, $A = \phi(\ex(A))$ holds by Proposition~\ref{prop:cg-meet-partition}.
Thus, the predecessors of $\phi(A)$ are precisely the sets $\phi(A) \setminus \{a\}$ for $a \in A$. 
Hence, for any $C \in \C$ such that $C \subset \phi(A)$, there is $a \in A$ such that $C \subseteq \phi(A) \setminus \{a\}$.
As $x \in \ex(\phi(A) \setminus \{a\})$, $x \in C$ entails $x \in \ex(C)$.
We conclude that $\phi(A) \in \min_\subseteq \{C : C \in \C,\ x \in C,\ x \notin \ex(C)\}$.
\end{proof}
Proposition~\ref{prop:carac-crit} suggests an algorithm to compute a critical generator of an element $x$ given a set that generates $x$.

\begin{proposition} \label{prop:critical-poly}
Let $(X, \C)$ be a convex geometry, $x \in X$ and $Y \subseteq X$ such that $x \notin Y$ and $x \in \phi(Y)$.
A critical generator $A$ of $x$ such that $\phi(A) \subseteq \phi(Y)$ can be computed with a number of calls to $\phi$ being polynomial in $\card{X}$.
\end{proposition}

\begin{proof}
Notice that $x \in \phi(Y)$ but $x \notin \ex(\phi(Y))$ by definition.
Hence, there exists $C \in \C$ such that $C \subseteq \phi(Y)$ and $C \in \min_\subseteq \{C : C \in \C,\ x \in C,\ x \notin \ex(C)\}$.
Moreover, for each such $C \subseteq \phi(Y)$ and each $C' \in \C$ such that $C \subseteq C' \subseteq \phi(Y)$, $x \notin \ex(C')$ must hold.

We describe an algorithm to identify a critical generator $A$ of $x$ included in $\phi(Y)$.
As a preprocessing step, we compute $\phi(Y)$.
Then, given a closed set $C$, we compute its extreme points and its predecessors, that are precisely the sets $\{C \setminus \{y\} : y \in \ex(C)\}$, and we identify the 
element $y \in \ex(C)$ with the smallest index such that $x \notin \ex(C \setminus \{y\})$.
If no such $y$ exists, we output $\ex(C)$.
Otherwise, we repeat the step with $C \setminus \{y\}$ as input.
By Proposition~\ref{prop:carac-crit}, and since we check all the predecessors of the input closed set at each step, this procedure will correctly output a critical generator $A$ of $x$ included in $\phi(Y)$.

Computing $\phi(Y)$ requires one call to $\phi$.
For a closed set $C$, finding $\ex(C)$ and checking whether $x \in \ex(C \setminus \{y\})$ for each $y \in \ex(C)$ requires at most $2 \cdot \card{X}$ calls to $\phi$.
Since at each step we remove one element from $C$, we will repeat these operations at most $\card{X}$ times.
Hence, the whole algorithm requires at most $1 + 2 \cdot \card{X}^2$ calls to $\phi$.
\end{proof}

\myparagraph{Implicational bases.}
An \emph{implication} over $X$ is a statement $A \to B$ where $A, B$ are subsets of $X$.
In $A \to B$, $A$ is the \emph{premise} and $B$ the \emph{conclusion}.
If $B$ is a singleton, that is $B = \{b\}$, the implication $A \to b$ is \emph{(right-)unit}.
An \emph{implicational base} (IB) is a pair $(X, \Sigma)$ where $\Sigma$ is a collection of implications over $X$.
The IB $(X, \Sigma)$ is \emph{(right-)unit} if it consists in unit implications.
The \emph{aggregation} of an IB $(X, \Sigma)$ is the IB $(X, \Sigma^*)$ resulting from the unification of the implications with equal premises, i.e., $\Sigma^* := \left\{A \to B : B = \bigcup \{B' : A \to B' \in \Sigma\}\right\}$.
Dually, the \emph{unit-expansion} is obtained by replacing each implication $A \to B \in \Sigma$ by the set of implications $\{A \to b : b \in B\}$.
Given an IB $(X, \Sigma)$, $\card{\Sigma}$ is the cardinal of $\Sigma$, that is the number of implications in $\Sigma$.
The \emph{size} of the IB is the value $|X|+|\Sigma|$.
The \emph{degree} of an element $x\in X$, denoted $\deg(x)$, is the number of implications of $\Sigma$ in which $x$ appears. 
We refine this notion by defining the \emph{premise-degree} of $x$ as $\pdeg(x):=|\{A : A \to B \in \Sigma,\, x\in A\}|$, and its \emph{conclusion-degree} as $\cdeg(x):=|\{A\to B : A \to B \in \Sigma,\, x \in B\}|$.
Then $\deg(x)\leq \pdeg(x)+\cdeg(x)$. 
Equality may not hold in general because an element can belong to both premise and the conclusion of an implication.
In this paper though, we avoid this degenerate case by always considering the premise and conclusion of an implication to be disjoint, which can be ensured in polynomial-time without loss of generality with respect to the closure system.
Moreover, we will mainly use the \emph{degree}, \emph{premise-degree}, and \emph{conclusion-degree} of $\Sigma$ defined by
\begin{align*}
    \pdeg(\Sigma) &:= \max_{x\in X} \pdeg(x),\\
    \cdeg(\Sigma) &:= \max_{x\in X}\cdeg(x),\\
    \deg(\Sigma) &:= \max_{x\in X} \deg(x).
\end{align*}
Note that $\deg(\Sigma)\leq \pdeg(\Sigma)+\cdeg(\Sigma)$ for any $\Sigma$.

\begin{example} \label{ex:prelim-IB}
Let $X = \{1, 2, 3, 4, 5\}$ and consider the IB $(X, \Sigma)$ given by:
\[ 
\Sigma = \{4 \to 1, 25 \to 3, 12 \to 3, 15 \to 23\}.
\]
We have $\pdeg(2) = 2$, $\cdeg(2) = 1$, $\pdeg(\Sigma) = \pdeg(2) = \pdeg(1) = 2$, and $\cdeg(\Sigma) = \cdeg(3) = 3$. 
\end{example}

We now relate IBs to closure systems.
Let $(X, \Sigma)$ be an IB.
A set $C \subseteq X$ is \emph{closed} for $(X, \Sigma)$ if for every implication $A \to B \in \Sigma$, $A \subseteq C$ implies $B \subseteq C$.
The pair $(X, \C)$ where $\C$ is the family of closed sets of $(X, \Sigma)$ is a closure system.
An IB $(X, \Sigma)$ thus induces a closure operator $\phi(\cdot)$ that can be computed in polynomial time using the \emph{forward chaining} algorithm.
This procedure, which will be also used in some of our proofs, starts from a set $Y$ and produces a sequence of sets $Y =: Y_0 \subseteq Y_1 \subseteq \dots \subseteq Y_k := \phi(Y)$ such that $Y_i := Y_{i-1} \cup \bigcup \{B : \text{$A \to B \in \Sigma$, $A \subseteq Y_{i-1}$, $B \nsubseteq Y_{i-1}$} \}$ for $1 \leq i \leq k$.
On the other hand, any closure system $(X, \C)$ can be represented by several \emph{equivalent} IBs.
An implication $A \to B$ is \emph{valid} or \emph{holds} in $(X, \C)$ if for all $C \in \C$, $A \subseteq C$ implies $B \subseteq C$.
Equivalently, $A \to B$ holds in $(X, \C)$ if and only if $B \subseteq \phi(A)$.
For instance, $\Sigma = \{A \to \phi(A) : A \subseteq X\}$ constitutes an IB for $(X, \C)$ as well as $\Sigma' = \{C \cup \{x\} \to \phi(C \cup \{x\}) : C \in \C,\, x \notin C\}$ (recall that $\emptyset \in \C$ by assumption).
Moreover an IB is always equivalent to its aggregation, and to its unit-expansion.
For example, the IB of Example~\ref{ex:prelim-IB} is an IB of the closure system of Figure~\ref{fig:prelim-example}.

The fact that a closure system $(X, \C)$ can be represented by several IBs has lead to the study of a broad range of different IBs and different measures on IBs, as witnessed by two surveys on the topic~\cite{wild2017joy, bertet2018lattices}.
Among possible IBs, the canonical base~\cite{guigues1986familles}, the canonical direct base~\cite{bertet2010multiple, wild1994theory} and the $D$-base~\cite{adaricheva2013ordered} have attracted much attention.
In this paper, we will be interested instead in some measures that allow to compare two equivalent IBs.
Namely, we say that an IB $(X, \Sigma)$, or simply $\Sigma$, is \emph{minimum} if for any IB $(X, \Sigma')$ equivalent to $(X, \Sigma)$, $\card{\Sigma} \leq \card{\Sigma'}$.
Similarly, if $(X, \Sigma)$ is unit, we say that it is \emph{unit-minimum} if it is cardinality-minimum among equivalent unit IBs.
If $(X, \Sigma)$ is an arbitrary IB, then a minimum IB $(X, \Sigma')$ can be obtained in polynomial time, see e.g.,~\cite{ maier1983theory, shock1986computing, guigues1986familles,wild1994theory}.
On the other hand, finding an equivalent unit-minimum IB is computationally hard~\cite{ausiello1986minimal}.
Thus, in general, computing a unit-minimum IB is harder than computing a minimum one.

The \emph{degree} of a closure system is the least integer $k$ such that it admits an implicational bases of degree $k$. The premise- and conclusion-degree are defined analogously.

\myparagraph{Translation task.} 
To conclude the preliminaries, we define the two enumeration problems we study in this paper.
These are the problems of translating between an IB and irreducible closed sets, two representations of a closure system.
Formally, the problems read as follows:

\begin{problemgen}
\problemtitle{Irreducible Closed Sets Enumeration (\icsenum{})}
\probleminput{An implicational base $(X, \Sigma)$.}
\problemquestion{The irreducible closed sets of the associated closure system.}
\end{problemgen}

\begin{problemgen}
\problemtitle{Unit-Minimum Implicational Base Generation (\mibgen{})}
\probleminput{A ground set $X$, the irreducible closed sets $\irr(\C)$ of a closure system $(X, \C)$.}
\problemquestion{A unit-minimum implicational base $(X, \Sigma)$ of $(X, \C)$.}
\end{problemgen}

We stress the fact that in both problems, only an implicit (and possibly compact) description of the closure system is provided. 
Namely, for \mibgen{}, only $X$ and $\irr(\C)$ are part of the input, not $\C$, and a unit-minimum implicational base is to be computed. 
Analogously for \icsenum{}, only $(X, \Sigma)$ is given, not $\C$.
Moreover, we note that in the literature, variants of \mibgen{} have been considered where the goal is to compute a minimum (not unit-minimum) IB.
However, and as discussed previously, our formulation is harder; see also~\cite{defrain2021translating}. 
It is also known that these problems are polynomially equivalent in general; see~\cite{khardon1995translating} for more details.
Finally, we note that in both problems the output may be of exponential size with respect to the input size as suggested by Example~\ref{ex:exponential} below.
Hence, these problems are studied through the lens of output-sensitive complexity.
\begin{example} \label{ex:exponential}
Let $X = \{a_1, \dots, a_k\} \cup \{b_1, \dots, b_k\} \cup \{x\}$, and $\Sigma_1 = \{a_i b_i \to x : 1 \leq i \leq k\}$. Denoting $(X, \C_1)$ the associated closure system, we have $\irr(\C_1) = \{X \setminus \{y\} : y \in X,\ y \neq x \} \cup \{ \{c_1, \dots, c_k\} : c_i \in \{a_i, b_i\},\ 1 \leq i \leq k\}$.
Therefore, $\card{\irr(\C_1)} = 2^k + 2k$ while $\card{\Sigma_1} = k$.

Dually, consider the closure system $(X, \C_2)$ with irreducible closed sets $\irr(\C_2) = \{X \setminus \{y\} : y \in X,\ x \neq y\} \cup \{X \setminus \{a_i, b_i, x\} : 1 \leq i \leq k\}$.
Then, a unit-minimum IB $(X, \Sigma_2)$ of $(X, \C_2)$ is given by $\Sigma_2 = \{ c_1 \dots c_k \to z : c_i \in \{a_i, b_i\}, \ 1 \leq i \leq k\}$.
We have $\card{\irr(\C_2)} = 3k$ while $\card{\Sigma_2} = 2^k$. 
The two IBs are pictured in Figure~\ref{fig:exponential}.
\begin{figure}[th!]
    \centering
    \includegraphics[scale=0.9]{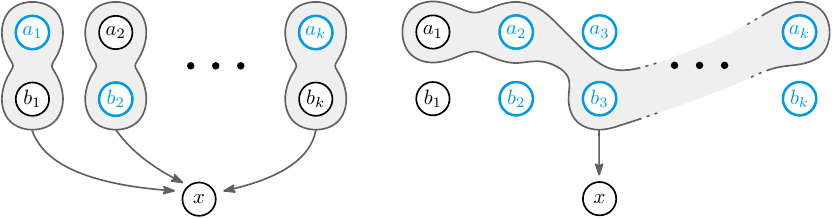}
    \caption{The implicational bases of Example \ref{ex:exponential} with $\Sigma_1$ on the left and $\Sigma_2$ on the right.
    For clarity we only give some of the implications, drawn in grey.
    In both cases, an irreducible closed set attached to $x$ is highlighted in blue. }
    \label{fig:exponential}
\end{figure}
\end{example}

As mentioned in the introduction, the two problems have been widely studied, but despite being harder than \misenum{} (hence \mhsenum{}), their complexity remains unknown to date.
In our journey to solve \icsenum{}, we will use the related problem of generating the family of irreducible closed sets attached to a given element:

\begin{problemgen}
\problemtitle{Attached Irreducible Closed Sets Enumeration (\acsenum{})}
\probleminput{An implicational base $(X, \Sigma)$ and $x \in X$.}
\problemquestion{The family $\irr(x)$ of irreducible closed sets attached to $x$.}
\end{problemgen}

It is well known that an algorithm solving \acsenum{} in output-polynomial time can be applied to each element $x$ to produce all irreducible closed sets of a closure system in output-polynomial time.
In this process, at most a linear number of repetitions are to be handled, and these repetitions do not occur in convex geometries by Proposition~\ref{prop:cg-meet-partition}.
Unfortunately, it was shown by Kavvadias et al.~\cite{kavvadias2000generating} in the language of Horn CNFs that the problem \acsenum{} admits no output-polynomial time algorithm in general unless $\textsf{P} = \textsf{NP}$.
Further results regarding \mibgen{}, \icsenum{} and \acsenum{} will be discussed in the next section, which is devoted to acyclic convex geometries, the class of closure systems at the core of our contribution.

\section{Acyclic convex geometries} \label{sec:acyclic}

In this section we define and give the main properties of the class of convex geometries we study in this paper: \emph{acyclic convex geometries}.
They have been studied under various names such as poset-type convex geometries \cite{santocanale2014lattices}, $G$-geometries \cite{wild1994theory}, or as acyclic Horn functions from the perspective of Horn CNFs \cite{hammer1995quasi, boros2009subclass}.

To define acyclic convex geometries, we will use IBs, thus following the line of propositional logic~\cite{hammer1995quasi}.
To do so, we need the concept of implication graph.
Let $(X, \Sigma)$ be an IB.
For the time being, $(X, \C)$ needs not be a convex geometry.
The \emph{implication graph} \cite{hammer1995quasi, boros2009subclass} of $(X, \Sigma)$ is the directed graph $G(\Sigma) = (X, E)$ where there is an arc $ab$ in $E$ if there exists $A \to B$ in $\Sigma$ such that $a \in A$ and $b \in B$.
The IB $(X, \Sigma)$ is called \emph{acyclic} if $G(\Sigma)$ is.
In this case, $G(\Sigma)$ can be seen as a poset, and naturally induces an \emph{ancestor} relation over $X$ defined for all $x \in X$ as follows:
\begin{align*}
\anc(x) & := \{y : \text{$y \in X$, $x \neq y$, and there exists a directed path from $y$ to $x$ in $G(\Sigma)$}\}. 
\end{align*}
Equivalently, $\anc(x)$ is the set of in-neighbors of $x$ in the transitive closure of $G(\Sigma)$.
The descendant relation can be defined dually using out-neighbors.
An element $x$ of $X$ is thus maximal if it has no ancestors, and minimal if it has no descendants.
As an example, Figure \ref{fig:prelim-graph} illustrates the implication graph of the IB given in Example \ref{ex:prelim-IB}.

We note that the notions of premise- and conclusion-degree differ from the in- and out-degree of the implication graph. 
Indeed, on the first hand, a single implication in $\Sigma$ may lead to several arcs in $G(\Sigma)$.
On the other hand, two implications with a same conclusion $y$ and sharing a same element $x$ in their premises yield a single arc $(x,y)$ in $G(\Sigma)$.

A closure system that admits an acyclic IB is \emph{acyclic} and it is a well-known fact that acyclic closure systems are convex geometries~\cite{wild1994theory, santocanale2014lattices}.
\begin{figure}[ht!]
    \centering
    \includegraphics[page=2, scale=1.1]{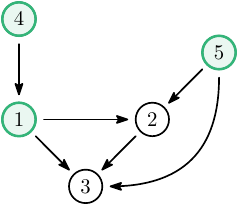}
    \caption{The implication graph $G(\Sigma)$ of the IB $(X, \Sigma)$ of Example~\ref{ex:prelim-IB}, associated to the closure system in Figure \ref{fig:prelim-example}.
    This directed graph is acyclic, which makes the closure system an acyclic convex geometry.
    The ancestors of $2$, highlighted in green, are $1, 4, 5$, i.e., $\anc(2) = \{1, 4, 5 \}$. 
    Maximal elements are $4$, $5$ and the unique minimal element is $3$.
    A \emph{topological order}, as we will use later, would be $4, 1, 5, 2, 3$ in this example.}
    \label{fig:prelim-graph}
\end{figure}

At first sight, acyclicity seems to be a property of IBs more than a property of the underlying closure system, as when taking two distinct IBs of a closure system $(X, \C)$, one may be acyclic and not the other. See Example~\ref{ex:acyclic-not-acyclic-IB}.

\begin{example}\label{ex:acyclic-not-acyclic-IB}
Let $X = \{1, 2, 3, 4\}$ and consider the two IBs $(X, \Sigma_1$), $(X, \Sigma_2)$ where:
\begin{itemize}
    \item $\Sigma_1 = \{1 \to 4, 12 \to 3, 23 \to 4\}$
    \item $\Sigma_2 = \{1 \to 4, 124 \to 3, 23 \to 4\}$
\end{itemize}
The two implicational bases are equivalent, i.e., they define the same closure system.
However, $G(\Sigma_1)$ is acyclic, while $G(\Sigma_2)$ contains a cycle with arcs $34$ and $43$.
\end{example}

Consequently, one might expect that acyclicity does not especially reflects in closed sets more than the fact that the closure system is a convex geometry.
This however is not true. 
To see this, consider the implicational base of minimal generators $(X, \Sigma_\mingen)$ of a closure system $(X, \C)$, where $\Sigma_\mingen = \{A \to b : b \in X,\, A \in \mingen(b)\}$, which has been studied under various names~\cite{bertet2010multiple}.
The implication graph $G(\Sigma_\mingen)$, where there is an arc from $a$ to $b$ if some $A \in \mingen(b)$ contains $a$, turns out to be equal to the (inverse of the) $\delta$-relation of lattice theory~\cite{monjardet1997dependence,bertet2010multiple} where $b \delta a$ holds in $(X, \C)$ if there exists $M \in \irr(b)$ such that $a \notin M$.
This equality is also a direct consequence of the fact that $\irr(b)$ and $\mingen(b)$ define dual hypergraphs, as explained in Section~\ref{sec:preliminaries}.
As was shown by Hammer and Kogan~\cite{hammer1995quasi}, $G(\Sigma_\mingen)$, or equivalently $\delta$, characterizes the acyclicity of a closure system.

\begin{lemma}[{\cite[Lemma 5.2]{hammer1995quasi}}] \label{lem:mingen-acyc}
A closure system $(X, \C)$ is acyclic if and only if the implication graph $G(\Sigma_\mingen)$ of its implicational base of minimal generators $(X, \Sigma_\mingen)$ is acyclic.
\end{lemma}
Lemma~\ref{lem:mingen-acyc} has interesting consequences.
First, given the correspondence between $G(\Sigma_\mingen)$ and the $\delta$-relation using irreducible closed sets, it shows that acyclicity indeed reflects in (irreducible) closed sets and can be tested in polynomial time given $\irr(\C)$.
Furthermore, it follows that $G(\Sigma_\mingen)$, and hence its associated ancestor relation, can be recovered in polynomial time from $\irr(\C)$ too.
To do so, one just has to check, for all $a, b \in X$, whether there exists $M \in \irr(b)$ such that $a \notin M$.
As another consequence, any subset $\Sigma$ of $\Sigma_\mingen$ produces an acyclic IB.
This holds in particular for IBs $(X, \Sigma)$ of $(X, \C)$ that satisfies $\Sigma \subseteq \Sigma_\mingen$.
Since such an IB $(X, \Sigma)$ can always be obtained from an arbitrary IB of $(X, \C)$ in polynomial time, checking whether $(X, \C)$ is indeed acyclic can be done in polynomial time given an arbitrary IB (see Corollary 5.3 in \cite{hammer1995quasi}).
Moreover, Boros et al.~\cite[Theorem 17, Corollary 18]{boros1998horn} show that any IB $(X, \Sigma)$ of $(X, \C)$ that satisfies $\Sigma \subseteq \Sigma_\mingen$ enjoys the same ancestor relation as $(X, \Sigma_\mingen)$.
In other words, $G(\Sigma)$ and $G(\Sigma_\mingen)$ have the same transitive closure.
We thus directly obtain a statement that will be used in a few paragraphs.
\begin{theorem} \label{thm:anc-poly}
Let $(X, \C)$ be an acyclic convex geometry and let $(X, \Sigma)$ be an implicational base of $(X, \C)$ such that $\Sigma \subseteq \Sigma_\mingen$.
Then, the ancestor relation associated to $\Sigma$ can be computed in polynomial time when given $X$ and $\irr(\C)$.
\end{theorem}

Among the minimal generators of $(X, \C)$, we now focus on critical generators, for they play a crucial role in acyclic convex geometries.
If $(X, \C)$ is an acyclic convex geometry, its \emph{critical base} $(X, \Sigma)$ of $(X, \C)$ is defined as follows:
\[
\Sigma_{\crit} := \{A \to b : \text{$b \in X$ and $A$ is a critical generator of $b$}\}.
\]
As a preliminary remark, $(X, \Sigma_\crit)$ must be acyclic by Lemma~\ref{lem:mingen-acyc}, since critical generators are minimal generators.
A nice result of Hammer and Kogan~\cite{hammer1995quasi}\footnote{In their terminology, an implication $A \to b$ where $A$ is a critical generator of $b$ is called and \emph{essential implicate} of the corresponding Horn CNF.} shows that not only $(X, \Sigma_\crit)$ is indeed a valid IB of $(X, \C)$, which is not true for convex geometries in general~\cite{wild1994theory, nakamura2013prime, adaricheva2013ordered}, 
but also that it is unit-minimum and can be obtained in polynomial time from any other IB of $(X, \C)$.
More precisely, they show:
\begin{theorem}[Corollaries 5.6, 5.7 in~\cite{hammer1995quasi}] \label{thm:crit-optim}
Let $(X, \C)$ be an acyclic convex geometry.
Then, its critical base $(X, \Sigma_\crit)$ is unit-minimum and its aggregation is minimum.
Moreover, it can be computed in polynomial time from any implicational base of $(X, \C)$.
\end{theorem}

As a consequence, the critical base of an acyclic convex geometry constitutes a valid output to the \mibgen{} problem. 
Therefore, our approach to \mibgen{} will consist in solving the following problem:
\begin{problemgen}
\problemtitle{Critical Base Generation (\critbase{})}
\probleminput{A ground set $X$, and the irreducible closed sets of an acyclic convex geometry $(X, \C)$.}
\problemquestion{The critical base of $(X, \C)$.}
\end{problemgen}
Given that $(X, \Sigma_\crit)$ is indeed a valid acyclic IB of $(X, \C)$ and $\Sigma_\crit \subseteq \Sigma_\mingen$, Theorem \ref{thm:anc-poly} 
automatically yields the following corollary that we will use in Section~\ref{sec:mibgen}.
\begin{corollary} \label{cor:anc-poly}
The ancestor relation of the critical base of an acyclic convex geometry $(X,\C)$ can be computed in polynomial time from $X$ and $\irr(\C)$.
\end{corollary}
In the rest of the paper, we will call $\critanc(x)$ the ancestor relation associated to $(X, \Sigma_\crit)$.
As a last property of critical generators, we argue that on top of being unit-minimum, the critical base of $(X, \C)$ is also minimum in terms of the different measures we introduced in Section~\ref{sec:preliminaries}: $\cdeg$, $\pdeg$, and $\deg$.
To do so, we will rely on a result stating that in a convex geometry $(X, \C)$, not necessarily acyclic, any critical generator is a subset of a premise in any IB of $(X, \C)$.
This is a consequence of \cite[Theorem 4.8]{korte1983shelling}.
It is also directly stated in \cite[Lemma 3.10]{yoshikawa2017representation}, as well as in \cite[Proposition 4.5]{defrain2021translating} but restricted to acyclic convex geometries.
In view of our claim on degrees though, we need to prove a slightly stronger property stating that these premise and critical generator share the same closure.
We give the statement below and we include for self-containment a proof which differs from the ones presented in~\cite{korte1983shelling,yoshikawa2017representation, defrain2021translating}.

\begin{lemma} \label{lem:crit-pseudo}
Let $(X, \C)$ be a convex geometry and let $b \in X$.
For any implicational base $(X, \Sigma)$ of $(X, \C)$ and any critical generator $A$ of $b$, there exists an implication $D \to B$ in $\Sigma$ such that $A \subseteq D$, $b \in B$ and $\phi(A) = \phi(D)$.
\end{lemma}

\begin{proof}
Let $(X, \Sigma)$ be an IB of $(X, \C)$, let $b \in X$ and $A$ a critical generator of $b$.
Let $A = A_0 \subseteq A_1 \subseteq \dots \subseteq A_k = \phi(A)$ be the sequence of sets obtained by applying the forward chaining on $A$ with $(X, \Sigma)$.
Since $A$ is a critical generator of $b$, $b \notin A$ but $b \in \phi(A)$.
Therefore, there exists some $1 \leq i \leq k$ such that $b \notin A_{i-1}$ but $b \in A_i$.
Hence, there exists an implication $D \to B \in \Sigma$ such that $b \in B$ and $D \subseteq A_{i-1}$.
Since $D \subseteq A_{i - 1}$ and $\phi(A_{i-1}) = \phi(A)$, we deduce $\phi(D) \subseteq \phi(A)$.
Now, since $A$ is a critical generator of $b$, we have that $\phi(A) \in \min_\subseteq \{C : C \in \C,\, b \in C,\, b \notin \ex(C)\}$ by Proposition~\ref{prop:carac-crit}.
On the other hand, we have $b \in \phi(D)$ but $b \notin \ex(\phi(D))$.
We deduce from the minimality of $\phi(A)$ that $\phi(A) = \phi(D)$.
As $(X, \C)$ is a convex geometry, the unique minimal spanning set of $\phi(D)$ is $\ex(\phi(D))=\ex(\phi(A))=A$ by Propositions \ref{prop:cg-meet-partition} and \ref{prop:carac-crit}.
Hence, $A \subseteq D$ holds.

Consequently, $D \to B$ is an implication of $\Sigma$ such that $A \subseteq D$, $\phi(D) = \phi(A)$ and $b \in B$ as desired.
This concludes the proof.
\end{proof}

We can now proceed to show that the critical base, once aggregated, minimizes all the three measures about degrees.

\begin{lemma} \label{lem:crit-degree}
Let $(X, \C)$ be an acyclic convex geometry with aggregated critical base $(X, \Sigma_\crit^*)$.
For any $\mu \in \{\cdeg, \pdeg, \deg\}$, we have:
\[ 
    \mu(\Sigma_\crit^*) = \min \{\mu(\Sigma) : \text{$(X, \Sigma)$ is an implicational base of $(X, \C)$}\}. 
\] 

\end{lemma}

\begin{proof}
We show the statement for each of the three measures.
We start with the equality regarding $\cdeg$, the conclusion-degree.
Let $(X, \Sigma)$ be an arbitrary IB of $(X, \C)$.
We show that $\cdeg(\Sigma_\crit^*) \leq \cdeg(\Sigma)$.
By definition of $\cdeg$, it is sufficient to argue that for each $b \in X$, $\cdeg_{\Sigma_\crit^*}(b) \leq \cdeg_\Sigma(b)$.
Let $b \in X$ and let $A_1, \dots, A_k$ be the critical generators of $b$.
Then $\cdeg_{\Sigma_\crit^*}(b) = k$.
By Lemma~\ref{lem:crit-pseudo}, for each $1 \leq i \leq k$ there exists an implication $D_i \to B_i$ in $\Sigma$ such that $b \in B_i$ and $\phi(A_i) = \phi(D_i)$.
By Propositions~\ref{prop:cg-meet-partition} and \ref{prop:carac-crit}, each $A_i$ is the unique minimal spanning set of $\phi(A_i)$.
Hence, $\phi(A_i) \neq \phi(A_j)$ for every $1 \leq i, j \leq k$ such that $i \neq j$.
Since for each $1 \leq i \leq k$, $\phi(D_i) = \phi(A_i)$, we deduce $\phi(D_i) \neq \phi(D_j)$ for each $j \neq i$.
Consequently, all implications $D_1 \to B_1, \dots, D_i \to B_i$ must be distinct, all the while satisfying $b \in B_i$ for each $1 \leq i \leq k$.
We deduce $k = \cdeg_{\Sigma_\crit^*}(b) \leq \cdeg_\Sigma(b)$.

We proceed to analyze $\pdeg$, the premise-degree.
Again, we need only to prove that for each $a \in X$, $\pdeg_{\Sigma_\crit^*}(a) \leq \pdeg_\Sigma(a)$.
Let $a \in X$ and let $A_1, \dots, A_k$ be the (distinct) critical generators to which $a$ belongs.
We obtain $\pdeg_{\Sigma_\crit^*}(a) = k$ because $\Sigma_\crit^*$ is aggregated.
Applying Lemma~\ref{lem:crit-pseudo} once more, there exists $D_1 \to B_1, \dots, D_k \to B_k$ in $\Sigma$ such that, for each $1 \leq i \leq k$, $a \in A_i \subseteq D_i$ and $\phi(A_i) = \phi(D_i)$.
Similarly to $\cdeg$, we obtain $\phi(D_i) \neq \phi(D_j)$ for all $1 \leq i, j \leq k$ such that $i \neq j$.
Thus, all implications $D_1 \to B_1, \dots, D_k \to B_k$ are distinct, with $a \in D_i$ for each $1 \leq i \leq k$.
We deduce $k = \pdeg_{\Sigma_\crit^*}(a) \leq \pdeg_\Sigma(a)$ for all $a \in X$, and $\pdeg(\Sigma_\crit^*) \leq \pdeg(\Sigma)$ follows.

As for the degree $\deg$, we have $\deg_{\Sigma_\crit^*}(x) = \cdeg_{\Sigma_\crit^*}(x) + \pdeg_{\Sigma_\crit^*}(x)$ for each $x \in X$ as $x$ cannot belong to one of its critical generators by definition.
According to the previous arguments, we thus have $\deg_{\Sigma_\crit^*}(x) \leq \deg_{\Sigma}(x)$ for each $x$.
So $\deg(\Sigma_\crit^*) \leq \deg(\Sigma)$ also holds, which concludes the proof.
\end{proof}

Moreover, observe that $\cdeg(\Sigma_\crit) = \cdeg(\Sigma_\crit^*)$ so the critical base, which is not aggregated, is enough for the conclusion degree. 
We directly deduce the following for acyclic convex geometries.

\begin{corollary} \label{cor:crit-degree}
An acyclic convex geometry has bounded (premise- or conclusion-) degree if and only if its aggregated critical base has bounded (premise- or conclusion-) degree.
\end{corollary}

To conclude this section, we give further insights on the translation tasks within acyclic convex geometries.
First, thanks to Theorem~\ref{thm:crit-optim}, one can always assume to be given the critical base as an input to \icsenum{} in this context.
More importantly, the critical base is, by Theorem~\ref{thm:crit-optim}, a correct solution to \mibgen{}, a fact that we will rely upon in Section~\ref{sec:mibgen}.

It is shown in \cite[Theorem~3.1]{defrain2021translating} that in acyclic convex geometries, all three problems \icsenum{}, \acsenum{} and \mibgen{} are harder than distributive lattice dualization, a generalization of \misenum{} and \mhsenum{}.
Here, we briefly discuss a simpler reduction that we will later adapt to obtain further hardness results for acyclic IBs with bounded degrees.
Let $\H = (V, \E)$ be a Sperner hypergraph as an input to \misenum{}.
Let $X := V \cup \{z\}$ where $z$ is a new gadget element and let $\Sigma := \{E \to z : E \in \E\}$.
The IB $(X, \Sigma)$ is our input to \icsenum{} and can be computed in polynomial time from $\H$.
It is acyclic, and furthermore it has height $2$, that is, $G(\Sigma)$ may be seen as a poset of height $2$.
The reduction is illustrated on an example in Figure \ref{fig:mis-to-ics}.
We proceed to describe the irreducible closed sets of $(X, \C)$, the closure system associated to $(X, \Sigma)$.
First, each element $v \in V$ is attached to a unique irreducible being $X \setminus \{v\}$, as $v$ is the conclusion of no implications.
As for $z$, by definition of $(X, \Sigma)$, $M \subseteq X$ belongs to $\irr(z)$ if and only if $M$ is inclusion-wise maximal for not including any premise of $\Sigma$, i.e., if and only if $M \in \MIS(\H)$.
Therefore, $\irr(\C) = \MIS(\H) \cup \{X \setminus \{v\} : v \in V\}$.
Given that $\irr(\C)$ has only a linear number of irreducible closed sets not in $\MIS(\H)$, which can be identified in polynomial time, an algorithm for \icsenum{} solves \misenum{} with the same time complexity.
Similarly, an algorithm solving \acsenum{} given $z$ would directly solve \misenum{}.

\begin{figure}[ht!]
    \centering
    \includegraphics[scale=1.0]{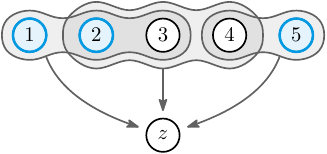}
    \caption{An illustration of the reduction from \misenum{} to \icsenum{} for Theorem~\ref{thm:ics-mib-height2-hardness}. The initial hypergraph is $\H = (\{1, 2, 3, 4, 5\}, \{123, 234, 45\})$ and the corresponding IB $(X, \Sigma)$ has implications $123 \to z, 234 \to z, 45 \to z$, drawn in grey.
    The set $125$, highlighted in blue, is a maximal independent set of $\H$ and hence an irreducible closed set attached to $z$.}
    \label{fig:mis-to-ics}
\end{figure}

For the \mibgen{} problem we instead use $\H$ has an input to \mhsenum{}.
We consider the closure system with $X := V \cup \{z\}$ and $\irr(\C) := \{V \setminus E : E \in \E\} \cup \{X \setminus \{v\} : v \in V\}$.
Note that $\irr(\C)$ can be built in polynomial time.
The unique unit-minimum IB associated to $(X, \C)$ is then $(X, \Sigma)$ with $\Sigma = \{T \to z : T \in \MHS(\H)\}$.
It is also acyclic and has height 2.

We conclude these observations with the following.

\begin{theorem}\label{thm:ics-mib-height2-hardness}
In acyclic convex geometries, the problems \icsenum{}, \acsenum{} and \mibgen{} are harder than \misenum{} (hence \mhsenum{}) even for acyclic convex geometries admitting an implicational base of height $2$.
\end{theorem}

By a small change in the construction of Theorem~\ref{thm:ics-mib-height2-hardness}, it can be shown that the result still holds for \acsenum{} in the case of bounded conclusion-degree.
Given $\H = (V, \E)$ with $\E = \{E_1, \dots E_m\}$, we put $X = V\cup\{z_i : 1\leq i\leq m\}$ where $z_1, \dots, z_m$ are new gadget elements and $ \Sigma = \{E_i \to z_i : 1\leq i\leq m\} \cup \{z_i \to z_{i+1} : 1\leq i <m\}$.
Note that $\Sigma$ has conclusion-degree 2.
The element $z_m$ will be the element as input of \acsenum{} along with $(X, \Sigma)$.
Again, we illustrate the reduction in Figure~\ref{fig:conclusion-hard}.
In the closure system associated to $(X, \Sigma)$, we have $\irr(z_m)= \MIS(\H)$.
We obtain the subsequent theorem.

\begin{theorem}\label{thm:conclusion-hard}
    The problem \acsenum{} is \misenum{}-hard even for acyclic implicational bases of conclusion-degree at most 2.
\end{theorem}

\begin{figure}[ht!]
    \centering
    \includegraphics[page=2, scale=1.0]{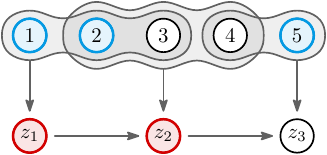}
    \caption{The reduction from \misenum{} to \acsenum{} with conlusion-degree at most $2$, applied to the hypergraph of Figure~\ref{fig:mis-to-ics}. 
    The implications are now $123 \to z_1, 234 \to z_2, 45 \to z_3, z_1 \to z_2, z_2 \to z_3$.
    With $z_1 \to z_2$, $z_2 \to z_3$, none of $z_1, z_2$, highlighted in red, cannot belong to any irreducible closed set attached to $z_3$.
    As compared to the reduction of Theorem~\ref{thm:ics-mib-height2-hardness}, irreducible closed sets of $z_3$ (formerly $z$) are unchanged.}
    \label{fig:conclusion-hard}
\end{figure}

Notice that this construction does not apply to \icsenum{}, as the $z_i$'s have non-trivial attached irreducible closed sets.
Indeed, for a given $i$, $\irr(z_i)$ is in a one-to-one correspondence with the maximal independent sets of the hypergraph consisting of the first $i$ edges $E_1, \dots, E_i$. 
Hence, the size of $\irr(\C)$ might be super-polynomial in the size of $\MIS(\H)$. 
Therefore, the question of the complexity of \icsenum{} for bounded conclusion-degree arises and will be dealt with in the following section.

\section{Top-down procedures}\label{sec:top-down}

We describe two sequential procedures 
to solve \icsenum{} and \critbase{} in acyclic convex geometries.
Intuitively, the first procedure consists in enumerating, incrementally and in a top-down fashion according to a topological ordering $x_1,\dots,x_n$ of the elements of the implicational base, all the irreducible closed sets attached to $x_i$.
The second one proceeds analogously and lists all critical generators of $x_i$ at each step. 

\subsection{Irreducible closed sets}\label{subsec:top-down-irr}
 
Let us focus on irreducible closed sets first.
Recall that the ancestors of an element $x$ in an acyclic implicational base are the elements that participate in a path of implications to~$x$, and which can be identified in polynomial time; see Section~\ref{sec:acyclic} for a formal definition.
We will be interested in solving the following version of \acsenum{} in which we are given ancestors' irreducible closed sets as additional information. 

\begin{problemgen}
  \problemtitle{Attached Closed Sets Enumeration with Ancestors' solutions (\acsaenum)}
  \probleminput{An acyclic implicational base $(X,\Sigma)$, an element ${x\in X}$, and the family $\irr(y)$ for every ancestor $y$ of $x$ in $\Sigma$.}
  \problemquestion{The set $\irr(x)$.}
\end{problemgen}

Clearly, an algorithm for \acsenum{} can be used to solve \acsaenum{}, while the reverse may not be true. 
We motivate this latter problem by showing that we have an incremental-polynomial time algorithm for \icsenum{} whenever we have one for \acsaenum.
Intuitively, our procedure will iteratively use an algorithm for \acsaenum{} on the elements of the ground set, in order, according to an arbitrary topological ordering of the implicational base.
Note that the base case corresponds to computing the irreducible closed set of the maximal elements $x$, i.e., those verifying $\irr(x)=\{X\setminus\{x\}\}$.
In terms of difficulty, \acsaenum{} may thus be seen as standing as an intermediate problem between \acsenum{} and \icsenum{} in acyclic convex geometries.

\begin{theorem}\label{thm:acsa-to-ics}
    Let $f\colon \mathbb{N}^2\to \mathbb{N}$ be a function.
    Let us assume that there is an algorithm that, given an instance of \acsaenum{} of size $N$, produces its $i^\text{th}$ solution in ${f(N,i)}$-time.
    Then there is an algorithm that, given an acyclic instance of \icsenum{} of size $N'$, produces its $j^\text{th}$ solution in $O(\poly(N') + N'\cdot f(N'+j,j))$ time.
\end{theorem}

\begin{proof}
    Let us assume the existence of an algorithm $\algoa$ for \acsaenum{} as required by the statement.
    Without loss of generality and for convenience in the rest of the proof, let us assume that $f$ is increasing.
    We describe a sequential procedure $\algob$ which produces the $j^\text{th}$ solution of \icsenum{} within the desired time bound.

    Let ${(X,\Sigma)}$ be an acyclic instance of \icsenum{}.
    Recall that, for convenience, we view the ancestor relation induced by $\Sigma$ as a partial order on $X$, and denote sources of $G(\Sigma)$ as maximal elements.
    First, we identify each maximal element $a$ of $X$, compute $\irr(a) = \{X \setminus \{a\}\}$ in $\poly(|X|, \card{\Sigma})$ time, and mark these $a$ as handled.
    Every such set is output and stored in memory.
    Then, iteratively, given a maximal element $x\in X$ among those that are not marked, we identify its ancestors in $\poly(|X|,|\Sigma|)$ time, output and store $\irr(x)$ using $\algoa$, and mark~$x$.
    The procedure stops when every element of $X$ has been marked.
    In other words, in this procedure, the elements $x$ are selected with respect to a topological ordering of $G(\Sigma)$.
    
    Note that we can indeed use $\algoa$ in the last step of our procedure since the sets $\irr(a)$ for all ancestors $a$ of $x$ have been inductively computed by the choice of~$x$.
    Furthermore, as by Proposition~\ref{prop:cg-meet-partition} the sets $\irr(a)$ and $\irr(b)$ are disjoint for any two distinct $a,b\in X$, there is no repetition in the outputs of this procedure: each set output by $\algoa$ is a distinct new solution of~$\algob$.

    We now turn to the complexity analysis.
    Observe that, except for the executions of~$\algoa$, the different steps of $\algob$ take $\poly(N')$ time for $N'=|X|+|\Sigma|$.
    Let us now focus on the time spent by $\algob$ in order to output the $j^\text{th}$ solution in case it is not the irreducible closed set of a maximal element, i.e., when it has been provided by an instance of $\algoa$.
    By assumption, this instance produces the $i^\text{th}$ irreducible closed set attached to $x$, each time being a new solution of $\algob$, in a time which is bounded by $f(N,i)$, for $N$ being the sum of the sizes of $X$, $\Sigma$ and $\bigcup_{a\in \anc(x)}\irr(a)$.
    Note that the latter family has already been output by $\algob$ at this stage and can be recovered in $\poly(j)$ time from the memory.
    Thus, and since $f$ is increasing, this is bounded by $f(|X|+|\Sigma|+j, j)$ for $j$ being the total number of solutions output by $\algob$ so far.
    And finally, $\algoa$ will be call only once for each $x$ of the ground set.
    We conclude to the desired time bound by replacing $|X|+|\Sigma|$ by $N'$.
\end{proof}

We note that an analogous statement to the one in Theorem~\ref{thm:acsa-to-ics} holds by relaxing the required and obtained time bounds to output-polynomial times.
On the other hand, improving the complexity of an algorithm for \acsaenum{} to polynomial delay (or even input-polynomial as we will provide later for the case of bounded conclusion-degree) does not generally ensure polynomial delay (or incremental-linear) for \icsenum{}, since the size of an instance of \acsaenum{} grows with the number of already obtained solutions for \icsenum{}, which typically becomes exponential in the size of the implicational base.

\subsection{Critical generators}\label{subsec:top-down-crit}

We now turn on enumerating the critical generators of an acyclic convex geometry.
Analogously as for irreducible closed sets, we proceed in a top-down fashion according to a topological ordering of the ground set, this time related to the critical base we wish to compute. 
Only it should be argued that such an ordering can be efficiently computed from the irreducible closed sets, i.e., without the knowledge of the critical base and its implication graph, a fact that follows from Corollary~\ref{cor:anc-poly}.

More formally, we are interested in solving the following enumeration problem, which is an analogue of \acsaenum{} to critical generators. 

\begin{problemgen}
  \problemtitle{Critical Generators of an Element with Ancestors' solutions (\cgeagen)}
  \probleminput{A ground set $X$, the irreducible closed sets $\irr(\C)$ of an acyclic convex geometry $(X,\C)$, $x\in X$, and the set $\critgen(y)$ for every $y\in \critanc(x)$.}
  \problemquestion{The family $\critgen(x)$.}
\end{problemgen}

We stress the fact that $\C$ is not given in this problem; only $X$ and $\irr(\C)$ are given as a representation. 
Note that $\critgen(x)$ is empty for maximal elements of the $\critanc$ relation, which corresponds to the base case of our procedure.
Moreover, sets lying in $\critgen(y)$ and $\critgen(y')$ for distinct elements $y,y'\in X$ yield distinct implications of the (unit) critical base of $\C$.
Then, following a direct analogue of the proof of Theorem~\ref{thm:acsa-to-ics} 
we obtain the following result.

\begin{theorem}\label{thm:cgeagen-to-mib}
    Let $f\colon \mathbb{N}^2\to \mathbb{N}$ be a function.
    Let us assume that there is an algorithm that, given any instance of \cgeagen{} of size $N$, produces its $i^\text{th}$ solution in ${f(N,i)}$-time.
    Then there is an algorithm that, given any acyclic instance of \critbase{} of size $N'$, produces its $j^\text{th}$ solution in $O(\poly(N') + N'\cdot f(N'+j\cdot |X|,j))$ time.
\end{theorem}

\subsection{Limitations}

Let us end this section by pointing some limitation to the top-down approach.
Note that for implicational bases (or critical base) of height at most two, the inputs of \acsaenum{} and \cgeagen{} do not contain any valuable additional information compared to the inputs of \icsenum{} and \critbase{}.
Indeed, as mentioned earlier, for such bases, each maximal element $a$ has a unique, polynomially characterized attached irreducible closed set $X \setminus \{a\}$, due to the fact that $a$ has no minimal generators, and in particular no critical generators.
Nevertheless, these instances remain \misenum{}-hard by Theorem~\ref{thm:ics-mib-height2-hardness}.
Thus, in order to expect tractable algorithms using this strategy, one should most probably restrict the closure system to particular cases or bound structural parameters. 

\section{Bounded conclusion-degree}\label{sec:conclusion-degree}

We show that the top-down approach successfully provides an incremental-polynomial time algorithm for \icsenum{} whenever the input implicational base has bounded conclu\-sion-degree, proving one of the two cases of Theorem~\ref{thm:icsenum}.

In the following, let us consider an instance of \acsaenum{}, i.e., we fix some acyclic implicational base $(X,\Sigma)$, an element $x\in X$, and assume that we are given $\irr(y)$ for all ancestors $y$ of $x$ in $\Sigma$.
Let us consider the hypergraph \emph{of premises of $x$} defined by
\begin{align*}
    \E_x &:= \{A: A\rightarrow B \in \Sigma, \ x\in B\},\ \text{and}\\
    \H_x &:= \left(\bigcup \E_x, \E_x\right)\! . 
\end{align*}
The intuition behind the following lemma is that, in order to get an irreducible closed set attached to $x$, one should remove at least one element $a$ in the premise of each implication having $x$ in the conclusion, which can be done by considering the intersection with an irreducible closed set attached to $a$.
Recall that $\S(T)$ denotes the family of irreducible selections; see Section~\ref{sec:preliminaries}.

\begin{lemma}\label{lem:conclusion-intersection-characterization}
For every ${M\in \irr(x)}$ there exists a minimal transversal $T$ of $\H_x$ and an irreducible selection $\S\in \S(T)$ such that 
\[
    M = \left(\bigcap \S \right) \setminus \{x\}.
\]
\end{lemma}

\begin{proof}
    Let $Y := \bigcup \E_x\setminus M$. 
    Note that for each $A\in \E_x$ there exists $a\in A$ such that $a\in Y$, as otherwise $A\subseteq M$ which contradicts the fact that $M \in \irr(x)$ as $A$ implies $x$. 
    Thus, $Y$ is a transversal of $\H_x$.
    Let $T\subseteq Y$ be any minimal transversal of $\H_x$ included in $Y$.
    Note that for all $a \in T$, as $M$ is a closed set and $a \notin M$, there exists a set $M_a \in \irr(a)$ such that $M \subseteq M_a$. 
    This defines an irreducible selection $\S=(M_a)_{a\in T}$ for $T$, where each such $M_a$ includes $M$. 
    Therefore, and as $x\notin M$, we have 
    $M \subseteq (\bigcap \S) \setminus \{x\}$.
    
    Now, as every member of $\S$ is closed, so is $\bigcap \S$.  
    Moreover, by definition of $\S$, no premise $A$ of $\E_x$ is contained in such an intersection. 
    Thus the set $M':=(\bigcap \S) \setminus \{x\}$ is closed.
    By the maximality of $M \in \irr(x)$, and since $M\subseteq M'$, we conclude that $M'=M$, hence to the desired equality.
\end{proof}

Note that, by definition, the number of edges in $\H_x$ is exactly $\cdeg(x)$, which is at most $\cdeg(\Sigma)$. 
Moreover, since the ground set of $\H_x$ is a subset of the ancestors of $x$, the selections of Lemma~\ref{lem:conclusion-intersection-characterization} are part of the input we consider for \acsaenum{}.
We use these observations to derive a polynomial-time algorithm for \acsaenum{} whenever the conclusion-degree of $x$ is bounded by a constant.

\begin{theorem}\label{thm:acs-a-input-poly}
    There is an algorithm solving \acsaenum{} in $N^{O(k)}$ time on inputs of size $N$ and maximum conclusion-degree $k$.
\end{theorem}

\begin{proof}
    As $\E_x$ has cardinality at most $k$, all possible minimal transversals $T$ of $\H_x$ can be computed in $|X|^{O(k)}$ time by brute forcing all subsets of size at most $k$, and testing in $\poly(|X|+|\H_x|)$ time, hence $\poly(N)$ time, whether the obtained set is a minimal transversal of $\H_x$.
    Moreover, all possible irreducible selections $\S\in \S(T)$ for each such $T$ can be computed in $(\max_{a\in T} |\irr(a)|)^{O(k)}$, hence in $N^{O(k)}$ time, since every $a\in T\subseteq \bigcup \E_x$ is an ancestor of $x$, thus the families $\irr(a)$ are part of the input.
    As the intersections of Lemma~\ref{lem:conclusion-intersection-characterization} can be computed in $\poly(N)$ time, and as it can be checked in $\poly(N)$ time if they are closed sets attached to $x$, we conclude that all $M\in \irr(x)$ can be identified in $N^{O(k)}$ time, as desired.
\end{proof}

We conclude to the following, as a corollary of Theorems \ref{thm:acsa-to-ics} and \ref{thm:acs-a-input-poly}, and proving the case of bounded-conclusion degree in Theorem~\ref{thm:icsenum}.

\begin{theorem}\label{thm:icsenum-conclusion}
    There is an incremental-polynomial time algorithm solving \icsenum{} in acyclic implicational bases of bounded conclusion-degree. 
\end{theorem}

Let us remark that in our algorithm, we actually do not need the knowledge of $\irr(a)$ for \emph{all} ancestors of $x$, but only for those in a premise which has $x$ in conclusion. 
This however makes no asymptotic difference on the final complexity we get for \icsenum{}.

Also, note that in the proof of Theorem~\ref{thm:acs-a-input-poly} we only need a constant upper bound on the size $k$ of a minimal transversal of $\H_x$, to derive an input-polynomial time algorithm solving \acsaenum{}.
We immediately derive the following generalization of Theorem~\ref{thm:icsenum-conclusion} dealing with hypergraphs of bounded dual-dimension (i.e., admitting minimal transversals of bounded size) who define complement of conformal hypergraphs~\cite{berge1984hypergraphs}, which typically arise in geometric settings.

\begin{theorem}\label{thm:bounded-conclusion}
    There is an incremental-polynomial time algorithm solving \icsenum{} if, for any element $x$ of the ground set, the hypergraph of premises having $x$ in their conclusion has bounded dual-dimension.
\end{theorem}

We end the section by showing the limitations of this approach in order to get output-polynomial time bounds whenever we drop the condition of bounded conclusion-degree.

First, note that the key step of the approach is to guess the correct $T\in \MHS(\H_x)$ and $\S\in \S(T)$ as in the statement of Lemma~\ref{lem:conclusion-intersection-characterization}, which we do by testing all possibilities in Theorem~\ref{thm:acs-a-input-poly}.
However, there are cases where the number of such possibilities is exponential in $|\irr(\Sigma)|$, which typically occurs if minimal transversals of $\H_x$ are of unbounded size, or if $\MHS(\H_x)$ is of exponential size in $\irr(x)$.

To get one such example, consider $k \in \mathbb{N}$ and the implicational base defined by $X:=\{a_1,b_1,\dots,a_k,b_k, x\}$ and $\Sigma:=\{a_ib_i \to x : 1\leq i \leq k\} \cup \{b_i \to b_{i+1}: 1 \leq i <k \}$; see Figure \ref{fig:dino-modif} for an illustration. 
Then $\irr(x) = \{ 
\{a_1,\dots,a_j\} \cup \{b_{j+1},\dots,b_k\} 
: 0\leq j \leq k \}$ so $|\irr(x)| = k+1$, and the other elements in $X\setminus \{x\}$ are easily seen to have a single attached irreducible closed set. Therefore, $\irr(\Sigma)$ is of polynomial size in $|X|=2k+1$ and $|\Sigma|= 2k-1$. 
However, $|\MHS(\H_x)| = 2^k$ and computing this set would fail in providing an output-polynomial time algorithm for enumerating $\irr(x)$. 

\begin{figure}[ht!]
    \centering
    \includegraphics[scale=1.0]{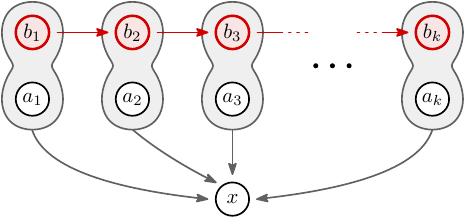}
    \caption{An instance where $|\MHS(\H_x)| = 2^k$ but $|\irr(\Sigma)| = 3k+1$ and $|\Sigma|=2k-1$.}
    \label{fig:dino-modif}
\end{figure}

At last, we note that this example has bounded premise-degree.
We will show in the next section that such instances can actually be efficiently solved by solution graph traversal.

\section{Bounded premise-degree}\label{sec:premise-degree}

We provide an incremental-polynomial time algorithm solving \icsenum{} whenever the given implicational base has bounded premise-degree.
Our algorithm relies on the top-down approach from Section \ref{sec:top-down} and solves \acsaenum{} relying on the solution graph traversal technique.

In the \emph{solution graph traversal} framework, one has to define a digraph called \emph{solution graph}, whose nodes are the solutions we aim to enumerate, and whose arcs $(F, F')$ consists of ways of reconfigurating a given solution $F$ into another solution $F'$.
Then, the goal is not to build the solution graph entirely, but rather to traverse all its nodes efficiently and avoiding repetitions of nodes.
This general technique has been fruitfully applied to numerous enumeration problems, see e.g., \cite{boros2004algorithms,khachiyan2008enumerating,yu2022efficient}, and generalized to powerful frameworks~\cite{avis1996reverse,cohen2008generating,conte2022proximity}.
In the following we will use the next folklore result which states conditions for such a technique to be efficient; see \cite{elbassioni2015polynomial} for a quantitative statement and a proof.

\begin{theorem}[Folklore]\label{thm:solution-graph}
    Let $\F\subseteq \pow{X}$ be a solution set (to be enumerated) over a ground set $X$, and let us denote by $N$ the size of the input.
    Then there is a $\poly(N)$-delay algorithm enumerating $\F$ whenever:
    \begin{enumerate}
        \item\label{it:initial} one member of $\F$ can be found in $\poly(N)$-time
        \item\label{it:function} there is a function $\N: \F\to \pow{\F}$ called \emph{neighboring function} mapping every solution to a set of solutions that can be computed in $\poly(N)$ time given $F \in \F$
        \item\label{it:connected} the directed solution graph $(\F, \{(F,F') : F'\in \N(F)\})$ is strongly connected
    \end{enumerate}
\end{theorem}

In the following, let us consider an instance of \acsaenum{}, i.e., we fix some acyclic implicational base $(X,\Sigma)$, an element $x\in X$, and we assume that we are given $\irr(y)$ for all ancestors $y$ of $x$.

We first argue that Condition~\ref{it:initial} holds for \acsaenum{}.
Note that there is a function $\GC_x$ which maps any $Y\subseteq X$ such that $x \notin \phi(Y)$ to a maximal set $M\supseteq Y$ with this property.
We call the obtained set the \emph{greedy completion} of $Y$.
Note that $\GC_x(Y)$ can be computed in polynomial time in the size of $(X,\Sigma)$, and that the obtained set lies in $\irr(x)$. 
We derive the following.

\begin{lemma}\label{lem:sol:initial}
    A first solution $\GC_x(\emptyset)\in \irr(x)$ can be obtained in $\poly(|X|+|\Sigma|)$ time.
\end{lemma}

\def\EMy{\E_{M,y}} 
\def\VMy{V_{M,y}}
\def\HMy{\H_{M,y}}

We now turn on showing the existence of a polynomial-time computable neighboring function satisfying Conditions \ref{it:function} and \ref{it:connected}.
Given some irreducible closed set $M\in \irr(x)$, and some element $y \in X\setminus M$ 
we define
\begin{align*}
    \EMy & := \{A : A\rightarrow B \in \Sigma\ \text{with}\ A\setminus M=\{y\}, \ B\not\subseteq M\},\\
    \VMy &:= \bigcup \EMy,\\
    \HMy &:= (\VMy,\EMy).
\end{align*}
In other words, $\HMy$ is the hypergraph of premises of $\Sigma$ not included in $M$, triggered by adding $y$, and with a conclusion not in $M$.
Note that $|\EMy| \leq \pdeg(y) \leq \pdeg(\Sigma)$.
Also, because every premise in $\EMy$ implies a $b \notin M$, hence an ancestor of $x$, we derive the following observation by transitivity.

\begin{observation}
Every vertex of $\HMy$ is an ancestor of $x$. 
\end{observation}

Regarding Condition \ref{it:function}, we first need the following statement, which holds the same idea as Lemma \ref{lem:conclusion-intersection-characterization} for bounded premise-degree.

\begin{lemma}\label{lem:intersection-progress}
    Let $M, M^*$ be two distinct sets in $\irr(x)$. 
    Then there exist $y\in M^* \setminus M$, $T\in \MHS(\HMy)$, and $\S\in \S(T)$ such that 
    \[
        M^* \subseteq \bigcap \S.
    \]
\end{lemma}

\begin{proof}
    Let $y$ be an element of $M^* \setminus M$ which is minimal, with respect to the ancestor order, for this property. 
    Let $A \in \EMy$ and $A \to B \in \Sigma$ such that $B \not \subseteq M$. As $y \in A$, every $b \in B$ is a descendant of $y$ so $b \notin M^*$ or $ b \in M$ by minimality of $y$.
    Moreover, $B \not \subseteq M$ by definition of $\EMy$ so there is a $b \in B$ such that $ b \notin M^*$.

    Therefore, for all $A \in \EMy$, $A \not \subseteq M^*$ ($M^*$ is a closed set). So $\VMy\cap M^*$ is an independent set of $\HMy$ thus its complement is a transversal. 
    Let $T$ be any minimal transversal of $\HMy$ included in $\VMy \setminus M^*$.
    Then, for all $a \in T$, $a \notin M^*$ so there exists $M_a \in \irr(a)$ such that $M^* \subseteq M_a$. For the irreducible selection $\S=(M_a)_{a\in T}$, one gets $M^* \subseteq \bigcap \S$.
\end{proof}


We are now ready to define the aforementioned neighboring function.
Given $M \in \irr(x)$ and $y \notin M$, we define
\begin{align*}
    \mathcal{N}(M,y) &:= \left\{\, \GC_x\left( \left(M \cap \bigcap \S \right)\cup \{y\} \right) :
    \begin{tabular}{l}
         $\S\in \S(T) \text{ and} \ T\in \MHS(\HMy)$ \\
         $\text{such that } x \notin \phi\big((M \cap \bigcap \S) \cup \{y\}\big)$
    \end{tabular}
    \right\}
\end{align*}
and let $\mathcal{N}(M) := \bigcup_{y\not\in M} \mathcal{N}(M,y)$.
We prove that Conditions~\ref{it:function} and \ref{it:connected} hold for this function.%

\begin{observation}\label{obs:sol:neighbor-codomain}
    By definition, the greedy completion in the definition of $\mathcal{N}(M,y)$ is well defined and $\mathcal{N}(M)\subseteq \irr(x)$.
\end{observation}

\begin{lemma}\label{lem:sol:neighbor-computation}
    Let $N$ be the sum of the sizes of $X$, $\Sigma$, and $\irr(a)$ for all ancestors $a$ of $x$.
    Then, for all $M \in \irr(x)$, $\mathcal{N}(M)$ can be computed in $N^{O(k)}$ time where $k$ is the premise-degree of $\Sigma$.
\end{lemma}

\begin{proof}
    Let us first note that $\GC_x\left( \left(M \cap \bigcap \S \right)\cup \{y\} \right)$ can be computed in $\poly(|X|+|\Sigma|+|\bigcap \S|)$ time, for any $y$ and $\S$. 
    Note that $\EMy$ has cardinality at most $k$, hence that all possible minimal transversals $T$ of $\HMy$ can be computed in $|X|^{O(k)}$ time by brute forcing all subsets of size at most $k$, and testing whether the obtained set is a minimal transversal of $\HMy$ in $\poly(|X|+|\HMy|)$ time, hence in $\poly(|X|+|\Sigma|)$ time.
    Moreover, all possible irreducible selections $\S\in \S(T)$ for each such $T$ can be computed in $(\max_{a\in T} |\irr(a)|)^{O(k)}$ time, hence in $N^{O(k)}$ time as families $\irr(a)$ are part of the input.
    Finally, the last union over $y$ multiplies the total time by a factor at most $|X| \leq N$, hence the result.
\end{proof}

Note that Condition~\ref{it:function} is ensured by Observation~\ref{obs:sol:neighbor-codomain} and Lemma~\ref{lem:sol:neighbor-computation}.
We are thus left to prove Condition~\ref{it:connected}.

\begin{lemma}\label{lem:sol:neighbor-proximity}
    Let $M, M^*$ be two distinct sets in $\irr(x)$.
    Then there exists $M'\in \mathcal{N}(M)$ such that $M\cap M^*\subset M'\cap M^*$. 
\end{lemma}

\begin{proof}
    Let $y \in M^*\setminus M$, $T \in \MHS(\HMy)$ and $\S \in \S(T)$ such that $M^* \subseteq \bigcap \S$, as provided in Lemma \ref{lem:intersection-progress}.
    Note that $T$ may intersect $M\setminus M^*$, but not $M^*$.
    Then 
    \begin{align*}
        M \cap M^* &\subseteq M \cap \bigcap \S\\
        &\subseteq \left( M \cap \bigcap \S \right) \cup \{y\}\\
        &\subseteq \GC_x\left( \left(M \cap \bigcap \S \right)\cup \{y\} \right) \in \mathcal{N}(M).
    \end{align*}
    So, by setting $M' := \GC_x\left( \left(M \cap \bigcap \S \right)\cup \{y\} \right)$ and intersecting with $M^*$, one gets $M\cap M^*\subseteq M'\cap M^*$.

    Finally, let us recall that $y \notin M$ by definition but $y \in M' \cap M^*$ by construction and definition. Hence the inclusion is strict.
\end{proof}

We derive the following, namely Condition~\ref{it:connected}, which is easily proved by induction since $M'\in \mathcal{N}(M)$ has greater intersection with $M^*$ than $M$ has with $M^*$.

\begin{corollary}\label{cor:sol:neighbor-proximity}
    The solution digraph $(\irr(x), \{(M,M') : M'\in \N(M)\})$ is strongly connected.
\end{corollary}

Combining Theorem~\ref{thm:solution-graph}, Lemma \ref{lem:sol:initial}, 
Lemma \ref{lem:sol:neighbor-computation}, and Corollary~\ref{cor:sol:neighbor-proximity}, we conclude to a polynomial-delay algorithm for \acsaenum{} by solution graph traversal for acyclic implicational bases of bounded premise-degree. 
This, together with Theorem~\ref{thm:acsa-to-ics}, in turn implies the following theorem, proving the remaining case of bounded premise-degree in Theorem~\ref{thm:icsenum}.

\begin{theorem}\label{thm:icsenum-premise}
   There is an incremental-polynomial time algorithm solving \icsenum{} on acyclic implicational bases of bounded premise-degree. 
\end{theorem}

As in the previous section, notice that we only needed to bound the size of a minimal transversal of $\HMy$ to derive an input-polynomial time algorithm solving \acsaenum{}.
This gives the following result, analogous to Theorem~\ref{thm:bounded-conclusion}.

\begin{theorem}\label{thm:bounded-premise}
    There is an incremental-polynomial time algorithm solving \icsenum{} if, for any irreducible closed set $M$ and $y \notin M$, the hypergraph $\HMy$ has bounded dual-dimension.
\end{theorem}

\section{Generating the critical implicational base} \label{sec:mibgen}

We provide an algorithm which constructs the critical base of an acyclic convex geometry given by its irreducible closed sets, namely solving \critbase{}.
Our algorithm relies on the top-down approach of Section~\ref{subsec:top-down-crit}, and solves \cgeagen{} by a saturation procedure which either finds a missing critical generator, or conclude that all critical generators have been found.
It performs in polynomial time whenever the aggregated critical base of the closure system has bounded premise- or conclusion-degree, relying on our algorithms from Sections~\ref{sec:conclusion-degree} and \ref{sec:premise-degree}, proving Theorem~\ref{thm:critgen}. 

In the \emph{saturation technique}, the goal is to decide, starting from the empty set as an initial partial solution set, whether a solution is missing, and to produce one such missing solution if it is the case.
This analogue functional problem is usually referred to as the \emph{decision version with counter example} of the enumeration problem.
It is folklore and easily seen that there is an incremental-polynomial time algorithm for an enumeration problem whenever its decision version with counterexample can be solved in polynomial time; see e.g.,~\cite{strozecki2019survey} for more details and \cite{fredman1996complexity,boros2003algorithms,mary2019efficient} for typical applications of this technique.

More formally, in the remaining of the section, we are interested in solving the following decision version with counterexample of \cgeagen{}.
Let us recall that $\critanc$, the ancestor relation associated to the critical base, can be computed in polynomial time from the irreducible closed sets by Corollary~\ref{cor:anc-poly}.

\begin{problemgen}
  \problemtitle{\makecell[l]{Critical Generators of an Element with Ancestors' solutions,\\ Decision version with counterexample (\cgeadec)}}
  \probleminput{A ground set $X$, the irreducible closed sets $\irr(\C)$ of an acyclic convex geometry $(X,\C)$, an element $x\in X$, the set $\critgen(y)$ for every $y\in \critanc(x)$, and a family $\F\subseteq \critgen(x)$.}
  \problemquestion{Yes if $\F=\critgen(x)$, a set $A\in \critgen(x)\setminus \F$ otherwise.}
\end{problemgen}

In the following, let us consider an instance of this problem.
Let us put
\[
    \Sigma':=\left\{A\to x : A\in \F\right\}\cup \left\{A\to y : A\in \critgen(y),\ y\in\critanc(x)\right\}.
\]
Note that $\Sigma'$ defines a subset of the critical base of $(X,\C)$.
We will note $(X,\C')$ the closure system represented by $\Sigma'$, and $\phi'$ its associated closure operator.
Since we are now dealing with two distinct closure systems, we shall write $\irr_{\C}(x)$ and $\irr_{\C'}(x)$ instead of $\irr(x)$ to distinguish between the two closure systems, and 
we do analogously for $\critgen$. 

As $\Sigma' \subseteq \Sigma_\crit$, one can notice that $\C \subseteq \C'$. This little fact will be a key element in the following proofs.

To solve \cgeadec{}, our strategy is to find a set $M \in \irr_{\C'}(x) \setminus \irr_{\C}(x)$ using an algorithm for \acsaenum{}. 
If such a $M$ exists, we argue that it provides a counterexample $A\in \critgen(x)\setminus \F$ (Lemma \ref{lem:extract-gen-from-irr} alongside Proposition \ref{prop:critical-poly}). 
If not, we prove that $\F= \critgen(x)$ (Lemma~\ref{lem:gen-crits-finished}).

\begin{lemma}\label{lem:extract-gen-from-irr}
    If there exists $M\in \irr_{\C'}(x)$ such that $M\not\in \irr_{\C}(x)$, then there exists $A\in \critgen_{\C}(x)$ such that $A\subseteq M$ and $A\not\in \F$.
\end{lemma}

\begin{proof}
First, we argue that $M \notin \C$ and $x \in \phi(M)$.
Assume for contradiction that $M \in \C$.
Because $x \notin M$ and $M \notin \irr_\C(x)$ by assumption, there exists $y \neq x$ such that $M \cup \{y\} \in \C$, as $(X, \C)$ is a convex geometry.
However, as $\Sigma' \subseteq \Sigma_\crit$, we have $\C \subseteq \C'$.
Therefore, $M \cup \{y\} \in \C'$ holds, which contradicts $M \in \irr_{\C'}(x)$.
We deduce $M \notin \C$.
Again as $\C \subseteq \C'$, $\phi(M) \in \C'$ holds.
From $M \in \irr_{\C'}(x)$ and $M \subset \phi(M)$, we derive $x \in \phi(M)$ as required.

Since $x \notin M$ but $x \in \phi(M)$ and $(X, \Sigma_\crit)$ is acyclic, there exists an implication $A \to b \in \Sigma_\crit$ with $A \subseteq M$ but $b \notin M$ such that $A \cup \{b\} \subseteq \critanc_\C(x) \cup \{x\}$.
If $b \neq x$, $A \to b \in \Sigma'$ by definition of $\Sigma'$.
As $M \in \C'$, $A \subseteq M$ thus implies $b \in M$, a contradiction.
We deduce that $b = x$ must hold, which concludes the proof.
\end{proof}

\begin{lemma}\label{lem:gen-crits-finished}
    If $\irr_{\C'}(x)\subseteq \irr_{\C}(x)$, then $\irr_{\C'}(x) = \irr_{\C}(x)$ and so $\F= \critgen(x)$.
\end{lemma}

\begin{proof}
First, we prove that $\irr_\C'(x) = \irr_\C(x)$ indeed holds.
Let $M \in \irr_\C(x)$.
Then, $M \in \C'$ as $\C \subseteq \C'$ and there exists $M' \in \irr_{\C'}(x)$ such that $M \subseteq M'$.
As $\irr_{\C'}(x) \subseteq \irr_\C(x)$ and the members of $\irr_\C(x)$ are pairwise incomparable, we deduce $M = M'$ and $M \in \irr_{\C'}(x)$.

Then, we argue that $\F = \critgen_\C(x)$ by contradiction.
Recall that $\F \subseteq \critgen_\C(x)$ by definition of $\F$.
Hence, suppose that $\F \subset \critgen_\C(x)$ and let $A \in \critgen_\C(x) \setminus \F$.
We first show that $x \notin \phi'(A)$. 
Again, suppose for contradiction that $x \in \phi'(A)$. 
Then, applying forward chaining on $A$ with $\Sigma'$ will reach $x$.
Since $\Sigma' \subseteq \Sigma_\crit$ and $A \to x \notin \Sigma'$, applying the forward chaining on $A$ with $\Sigma_\crit \setminus \{A \to x\}$ will reach $x$ too.
Hence, the two IBs $(X, \Sigma_\crit)$ and $(X, \Sigma_\crit \setminus \{A \to x\})$ are equivalent.
However, this contradicts the fact that $\Sigma_\crit$ is unit-minimum by Theorem~\ref{thm:crit-optim}.
We conclude that $x \notin \phi'(A)$ must hold as expected.
Now, $x \notin \phi'(A)$ implies that there exists $M \in \irr_{\C'}(x)$ such that $\phi'(A) \subseteq M$.
By assumption, $M \in \irr_\C(x)$, so that $A \subseteq \phi(A) \subseteq \phi(M) = M$.
However, $x \in \phi(A)$ holds as $A \in \critgen_\C(x)$, which entails $x \in M$, a contradiction with $M \in \irr_\C(x)$.
We conclude that such $A$ cannot exist, and hence that $\F = \critgen_\C(x)$, which was to be proved.
\end{proof}

Before concluding to an algorithm for \cgeadec{} based on \acsaenum{}, there is one last subtlety to handle: we need to guarantee that for each $y \in \critanc(x)$, the irreducible closed sets $\irr(y)$ we provide as an input to \acsaenum{} along with $\Sigma'$ indeed correspond to the irreducible closed sets attached to $y$ in $(X, \C')$, i.e., that $\irr_{\C}(y) = \irr_{\C'}(y)$ for each such $y$.

\begin{lemma} \label{lem:same-irreducibles-ancestor}
For every $y \in \critanc(x)$, $\irr_{\C}(y) = \irr_{\C'}(y)$.
\end{lemma}

\begin{proof}
Because attached irreducible closed sets and minimal generators define dual hypergraphs, it is sufficient to show that $\mingen_\C(y) = \mingen_{\C'}(y)$.
Since $(X, \Sigma_\crit)$ is acyclic, so is $(X, \Sigma')$ and $\critanc_\C(y) = \anc_{\Sigma'}(y)$ follows by definition of $\Sigma'$.
We deduce that for every $A \in \mingen_\C(x) \cup \mingen_{\C'}(x)$, $A \subseteq \critanc_\C(x)$.
Again using acyclicity, for each $Z \subseteq \critanc_\C(y)$, $y \in \phi(Z)$ if and only if there exists a subset $\Sigma'' = \{A_1 \to b_1, \dots, A_k \to b_k\}$ of implications of $\Sigma_\crit$ where $A_i \cup \{b_i\} \subseteq \critanc_\C(y) \cup \{y\}$ for all $1 \leq i \leq k$ and $y \in \phi''(Z)$, where $\phi''$ is the closure operator associated to $(X, \Sigma'')$.
Given that any such $\Sigma''$ satisfies $\Sigma'' \subseteq \Sigma' \subseteq \Sigma_\crit$, we deduce that $y \in \phi(Z)$ if and only if $y \in \phi'(Z)$.
Therefore, $A \in \mingen_\C(x)$ if and only if $A \in \mingen_{\C'}(x)$, which concludes the proof.
\end{proof}

We conclude to the following as a consequence of Lemmas \ref{lem:extract-gen-from-irr} and~\ref{lem:gen-crits-finished}, and derive its corollary by the saturation technique.

\begin{theorem}\label{thm:icsenum-to-cgeadec}
     There is a polynomial-time algorithm solving \cgeadec{} in acyclic convex geometries whenever there is an incremental-polynomial time algorithm solving \acsaenum{} within the same class.
\end{theorem}

\begin{proof}
    Let us assume the existence of an incremental-polynomial time algorithm $\algoa$ solving \acsaenum{} on acyclic implicational bases, and let $p: \mathbb{N}^2 \to \mathbb{N}$ be the polynomial such that $\algoa$ outputs the $i^\text{th}$ solution in time ${p(N,i)}$ on input $(X,\Sigma)$, $x\in X$ and $\bigcup_{y\in \anc(x)} \irr(y)$, whose total size is $N$.
    Let us further consider an instance $(X,\, \irr(\C),\, x,\, \{\critgen(y):y\in \critanc(x)\},\, \F)$ of \cgeadec{}, and $\Sigma'$ as defined above and its aggregation $\Sigma'^*$.

    In order to solve \cgeadec, we start simulating $\algoa$ on input $(X,\Sigma'^*)$, $x$ and $\{ \irr_\C(y) : y \in \critanc(x) \}$. 
    Note that is indeed a valid input of \cgeadec{} as the ancestor relation induced by $\Sigma'^*$ on $\critanc(x)$ is $\critanc$ and $\irr_\C(y)=\irr_{\C'}(y)$ for all $y \in \critanc(x)$, according to Lemma \ref{lem:same-irreducibles-ancestor}.
    Each time $\algoa$ outputs a set $M \in \irr_{\C'}(x)$, we check whether $M \in \irr_\C(x)$. 
    
    If not, we stop $\algoa$: there exists $A \in \critgen(x) \setminus \F$ included in $M$ according to Lemma~\ref{lem:extract-gen-from-irr}. 
    Thanks to Proposition \ref{prop:critical-poly}, we can retrieve such a $A$ in polynomial time and return it as a counterexample to \cgeadec. 
    
    On the other hand, if $\algoa$ terminates only having output sets in $\irr_\C(x)$, we derive that $\irr_{\C'}(x) \subseteq \irr_\C(x)$. 
    According to Lemma \ref{lem:gen-crits-finished}, it means in particular that $\critgen(x)=\F$.
    Hence we can answer \texttt{yes} for \cgeadec.
   
    Note that our simulation of $\algoa$ may output at most $|\irr_{\C}(x)|+1$ sets before it either terminates, or produces a counterexample.
    By assumption this takes a total time of ${p(N, |\irr_{\C}(x)|+1)}$ which is polynomial in the input size of \cgeadec.
\end{proof}

\begin{remark}
    Notice that the hypothesis on incremental-polynomial time for \acsaenum{} is fundamental, as one could not guarantee that we get the counter example early enough if it performed only in output-polynomial time.
\end{remark}

\begin{corollary}\label{cor:icsenum-to-cgeagen}
    There is an incremental polynomial-time algorithm solving \cgeagen{} in acyclic convex geometries whenever there is an incremental polynomial-time algorithm solving \acsaenum{} within the same class.
\end{corollary}

Note that, in this paper, we did \emph{not} provide an incremental-polynomial time algorithm for \acsaenum{} \emph{in general}, but managed to do so for implicational bases having bounded conclusion- or premise-degree.
It thus remains to show that our approach still holds for such class of closure systems in order to prove Theorem~\ref{thm:critgen}, which we do now.

First, note that in the proof of Theorem~\ref{thm:icsenum-to-cgeadec}, given the irreducible closed sets of a closure system $\C$, we need the existence of an incremental-polynomial time algorithm $\algoa$ solving \acsaenum{} on instances $\Sigma'^*$ obtained by considering an aggregated subset of its critical base.
Thus, in order to derive an algorithm solving \cgeadec{} in acyclic closure systems of bounded (premise- or conclusion-) degree, we need to argue that such an aggregation $\Sigma'^*$ represents an acyclic closure system that also has bounded \mbox{(premise- or conclusion-)} degree.
This follows from the fact the aggregated critical base of an acyclic convex geometry with bounded degree has bounded degree by Corollary~\ref{cor:crit-degree}, and that any subset $\Sigma'$ of the critical base $\Sigma_\crit$ of an acyclic convex geometry defines an acyclic convex geometry $\C'$.

Now, concerning the complexity, we note that the number of critical generators is bounded by $d \cdot |X|$ if the aggregated critical base has (premise- or conclusion-) degree $d$.
In addition, any incremental-polynomial time algorithm performs in (input) polynomial time on an instance with polynomially many solutions.
We derive Theorem~\ref{thm:critgen}, that we restate here, as a corollary of Theorem~\ref{thm:icsenum} together with these observations, Lemma \ref{lem:crit-degree}, Corollary~\ref{cor:icsenum-to-cgeagen}, and Theorem~\ref{thm:cgeagen-to-mib}.

\restatethmmibgen*

\section{Discussion}\label{sec:discussion}

In this paper, we provided incremental-polynomial time algorithms for the two translation tasks \icsenum{} and \critbase{} in acyclic convex geometries of bounded (premise- or conclusion-) degree. 
We note that the delay of our algorithms is intrinsically not polynomial since our algorithm rely on a sequential top-down procedure which reduces the enumeration to an auxiliary enumeration problem whose instance grows with the total number of solutions output so far. 
This suggests the following question.

\begin{question}\label{qu:delay}
    Can \icsenum{} and \critbase{} be solved with polynomial delay for acyclic convex geometries of bounded degree?
\end{question}

Toward this direction, we now argue that the classic flashlight search framework may not be used in a straightforward way to obtain a positive answer to this question.
In the \emph{flashlight search technique}, the goal is to construct solutions one element at a time, according to a linear order $x_1,\dots,x_n$ on the ground set, and deciding at each step whether $x_i$ should or should not be included in the solution.
For this approach to be tractable, the framework requires to solve the so-called ``extension problem'' in polynomial time: given two disjoint subsets $K, F$ of the ground set, decide whether there exists a solution including $K$ and avoiding $F$.
We refer to e.g.,~\cite{boros2004algorithms,mary2019efficient} for a more detailed description of this folklore technique.
More formally, this technique reduces the enumeration to the following decision problem. 

\begin{problemdec}
  \problemtitle{Irreducible Closed Set Extension (\icsext)}
  \probleminput{An implicational base $(X,\Sigma)$ and $K,F\subseteq X$.}
  \problemquestion{Does there exist $M\in \irr(\C)$ such that $K\subseteq M$ and $M\cap F=\emptyset$?}
\end{problemdec}

We show that this problem is \NP-complete even for acyclic implicational bases of bounded degree, even for ${K=\emptyset}$, suggesting that this framework may not be used in a straightforward way to improve our results.

\begin{theorem}\label{thm:extension}
    The problem \icsext{} is \NP-complete even for acyclic implicational bases with degree, premise-degree, conclusion-degree and dimension (i.e., the maximum size of a premise) simultaneously bounded by 4, 4, 2 and 2, respectively.
\end{theorem}

\begin{proof}
    A polynomial-size certificate for \NP{} is a subset $M \subseteq X$ for which it is easily checked, in polynomial time, whether $M\in \irr(\C)$, $K\subseteq M$, and $M\cap F=\emptyset$. 
    
    Regarding completeness, it is known that SAT is \NP-complete even for instances where each clause contains at most $3$ literals, and each variable appears (negatively or positively) in at most $3$ clauses~\cite[Problem L01]{garey2002computers}.
    Let us consider an instance $\psi$ of this problem, with variables 
    $v_1,\dots,v_n$ and clauses $C_1,\dots,C_m$.
    For clause $C_j$, $1\leq j\leq m$, let us denote by $\ell_j^1,\dots,\ell_j^k$ its $k\leq 3$ literals, i.e., $\ell_j^i$ is either a variable or its negation.

    \begin{figure}[t]
        \centering
        \includegraphics[scale=1.0]{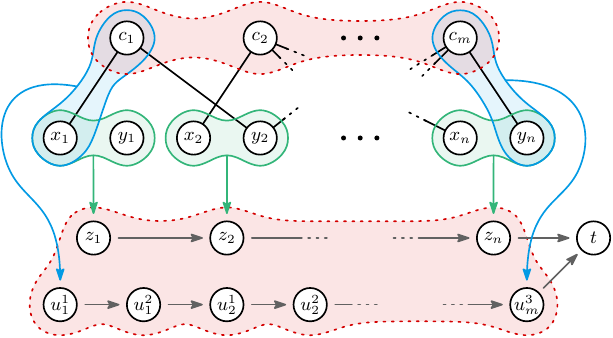}
        \caption{An illustration of the construction of Theorem~\ref{thm:extension}.
        Black edges correspond to the incidence bipartite graphs of $\psi$ and are represented for intuition only: they are not part of the \icsext{} instance. Rather, to each such edge corresponds a implication having an element of $U$ as conclusion; for readability, only two of these implications are represented.} 
        \label{fig:extension}
    \end{figure}

    We reduce $\psi$ to an instance $((X,\Sigma), \emptyset, F)$ of \icsext{} as follows.
    See Figure~\ref{fig:extension} for an illustration.
    Let us first describe the ground set $X$.
    First, for every variable $v_i$, $1\leq i\leq n$ we create three elements $x_i,y_i,z_i$ with $x_i$ representing $v_i$, $y_i$ representing $\overline{v_i}$, and $z_i$ being a gadget representing a conflicting selection of both $v_i$ and $\overline{v_i}$ in a non-valid assignment.
    Then, for every clause $C_j$, we create one element $c_j$ representing the clause, and $k=|C_j|\leq 3$ additional elements $u_j^1,\dots, u_j^k$ representing the literals of the clause.
    Finally, we add another gadget element $t$.

    We now describe $\Sigma$.
    For every $1\leq i\leq n$, we create an implication $\{x_i,y_i\} \to z_i$.
    For every clause $C_j$, and each variable $v_i$ appearing in position $k$ in $C_j$, $1\leq k\leq |C_j|$, we add the implication $\{c_j,x_i\} \to u_j^k$ if it appears positively, i.e., if $\ell_j^k=v_i$, and $\{c_j,y_i\} \to u_j^k$ if it appears negatively, i.e., if $\ell_j^k=\overline{v_i}$.
    Let $U := \bigcup_{1\leq j\leq m} \{ u_j^k: 1\leq k\leq |C_j|\}$ and $u_1,\dots,u_p$ be an arbitrary labeling of $U$.
    We complete the construction of $\Sigma$ by adding the implications $z_i\to z_{i+1}$ for every $1\leq i<n$, the implications $u_i\to u_{i+1}$ for every $1\leq i<p$, and the two implications $z_n\to t$ and $u_p\to t$. Notice that the premise of any implication is a minimal generator of $t$ by construction.

    Finally, we put $F=\{c_1,\dots,c_m\}\cup\{z_1,\dots,z_n\}\cup\{u_1,\dots,u_p\}$.

    Let us consider a positive instance of \icsext{}, i.e., some set $M\in \irr(\C)$ such that $M\cap F=\emptyset$.
    Let $a\in X$ such that $M\in \irr(a)$. As $z_n \in F$, it must be an ancestor of $a$. But note that $z_n$ implies a single element : $t$. Thus, $a=t$ hence $M\in \irr(a)$.
    

    Consider the assignment defined by
    \begin{equation*}
    v_i \mapsto 
    \begin{cases}
      1 & \text{if}\ x_i\in M, \\
      0 & \text{if}\ y_i\in M.
    \end{cases}
    \end{equation*}
    Note that this assignment is valid since $\{x_i,y_i\}$ is a minimal generator of $t$, hence not both $x_i$ and $y_i$ may belong to $M$.
    Let us argue that it satisfies $\psi$. 
    For every $1\leq j\leq m$, one of $\ell_j^k$, $1\leq k\leq |C_j|$ must be in $M$, as otherwise we can add $c_j$ to $M$ without implying $t$, a contradiction to $M\in \irr(t)$.
    Hence each clause is satisfied.
    This concludes the first direction.

    Conversely, suppose that $\psi$ is satisfiable.
    Consider a satisfying assignment and the set $M$ defined by taking $x_i$ if $v_i$ is set to true, and $y_i$ otherwise, for every $1\leq i\leq m$.
    Then $t\notin M$ and $M$ is closed, as it does not contain any premise. 
    Let us argue that it is maximal with that property, hence that $M\in \irr(t)\subseteq \irr(\C)$.
    Clearly, none of $z_1,\dots,z_n$ can be added to~$M$, and similarly for the elements of $U$ without adding $t$.
    Suppose toward a contradiction that an element $c_j$, $1\leq j\leq m$ can be added to $M$.
    Since $\psi$ is satisfiable, $M$ contains either $x_i$ or $y_i$ corresponding to a variable $v_i$ appearing either positively or negatively in $C_j$.
    Hence $M$ triggers one of the corresponding implication $\{c_j,\alpha\}\to u$ for $\alpha=x_i$ or $\alpha=y_i$ and $u\in U$.
    We derive that $M$ implies $t$, a contradiction.
    This concludes the proof of \NP{}-completeness.
    
    Lastly, notice that the implications in $\Sigma$ all have a premise of size 1 or 2 and that every element $x \in X$ is in at most 2 conclusions (exactly 2 for $t$, $z_i$ and $u_i$ with $i \geq 2$).
    Moreover, as a literal participates in at most 3 clauses, each of $x_i,y_i$, $1\leq i\leq n$ has total- and premise-degree at most 4, and the other elements have degree at most 3, hence the result.
\end{proof}

Among other directions, we wonder whether our results extend to broader classes of closure systems or to parameters that generalize the degree.
As a first step towards this avenue, we show that when we simply drop acyclicity, then the degree is helpless as far as \icsenum{} is concerned.
More precisely, we show in the next theorem that an algorithm solving \icsenum{} for closure systems with bounded degree can be used to solve \icsenum{} in the general case.

\begin{theorem} \label{thm:cycles-degree}
There is an output-polynomial time algorithm solving \icsenum{} in general whenever there is an output-polynomial time algorithm solving \icsenum{} on closure systems with bounded degree. 
\end{theorem}

\begin{proof}
Let $(X, \Sigma)$ be an instance of \icsenum{} with $X = \{x_1, \dots, x_n\}$ and $\Sigma = \{A_1 \to B_1, \dots, A_m \to B_m\}$.
We build an implicational base $(X', \Sigma')$ of polynomial size and that has bounded degree.
We first describe $X'$.
For each $x_i \in X$ and each implication $A_j \to B_j \in \Sigma$, we create a gadget element $x_i^j$ if $x_i \in A_j\cup B_j$ and we call $X_i$ the set of resulting elements associated with $x_i$, i.e., $X_i = \{x_i^j : A_j \to B_j \in \Sigma, \ x_i \in A_j \cup B_j\}$.
Recall that we consider that $B_j$ and $A_j$ are disjoint, so that an element can occur only once in an implication.
Therefore, $\card{X_i} = \deg(x_i)$.
We put $X' = \bigcup_{i = 1}^n X_i$.
We now describe $\Sigma'$.
For each implication $A_j \to B_j$ of $\Sigma$ we create an implication where each element $x_i \in A_j \cup B_j$ is replaced by its corresponding gadget $x_i^j$, that is, we build the implication $A'_j \to B'_j$ where $A'_j = \{x_i^j : x_i \in A_j\}$ and $B'_j = \{x_i^j : x_i \in B_j\}$.
Finally, for each $i$, we consider an arbitrary labeling $y_1, \dots ,y_k$ of the elements of $X_i$ and make all of them equivalent by putting in $\Sigma'$ the cycle of implications $\{y_1 \to y_2,\ \dots, \ y_{k-1} \to y_k, y_k \to y_1\}$.
This completes the description of $\Sigma'$.
Observe that $(X', \Sigma')$ can be built in polynomial time in the size of $(X, \Sigma)$ and that each element in $(X', \Sigma')$ has degree at most 3.
The reduction is illustrated on an example in Figure~\ref{fig:cycles}.

Note that the mapping $\psi \colon \C \to \C'$, $C \mapsto \bigcup_{x_i \in C} X_i$ is a polynomial-time computable bijection between the two closure systems $(X, \C)$ and $(X', \C')$.
Hence, with this bijection, an output-polynomial time algorithm solving \icsenum{} for $(X', \Sigma')$ can be used to solve \icsenum{} for $(X, \Sigma)$ within the same time bound.
\end{proof}

\begin{figure}[th!]
    \centering
    \includegraphics[scale=0.9]{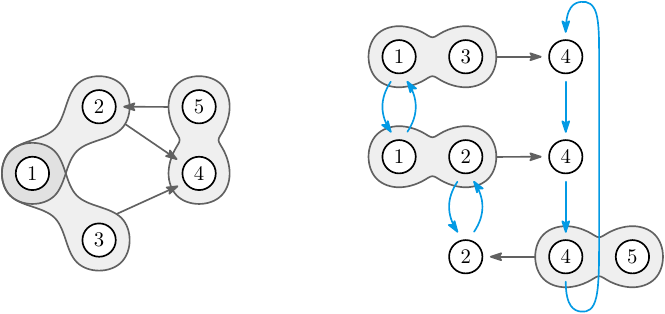}
    \caption{On the left, the IB $\Sigma = \{13 \to 4, 12 \to 4, 45 \to 2\}$. 
    On the right, we apply the reduction of Theorem~\ref{thm:cycles-degree}. 
    For readability, we repeat the label of each element in each implication, and we highlight in blue the implications of $\Sigma'$ encoding the equivalence between each repetition of an element.}
    \label{fig:cycles}
\end{figure}

Theorem~\ref{thm:cycles-degree} thus suggests to investigate instead more constrained classes of closure systems.
The top-down strategy we used for \icsenum{} requires to devise an ordering on the elements of the ground set that provides a meaningful ancestor relation.
Lower-bounded closure systems~\cite{freese1995free, adaricheva2013ordered} are a well-known generalization of acyclic convex geometries that appears as a natural candidate to investigate.
They are characterized by an acyclic $D$-relation, a subset of the $\delta$-relation (see Section~\ref{sec:acyclic}), which suggests that a topological ordering of the elements can be set and used as in our approach.
Besides, closure systems arising from graphs, also known as graph convexities~\cite{farber1986convexity}, may be a source of closure systems where natural orders may be found based on the properties on the underlying graphs.
We note, however, that these closure systems are in general not acyclic.

The algorithm we rely on to solve \mibgen{}, through \critbase{}, seems slightly more demanding. 
On top of a natural ordering of the vertices, it requires the critical base not only to be well defined, but also to be valid.
Critical generators are initially defined in convex geometries using extreme points, but in arbitrary closure systems, closed sets need not have extreme points, and critical generators may thus not exist.
However, as observed by different authors~\cite{nakamura2013prime, wild2017joy}, the critical generators of an element could instead be defined as its minimal generators that have the smallest closure inclusion-wise.
This definition makes critical generators of an element exist in any closure system as long as this element admits minimal generators, and it coincides with the original definition in convex geometries.
Nevertheless, even under this definition, the corresponding critical base of an arbitrary closure system needs not to be a valid IB (as it already fails in convex geometries), which prevents the strategy of finding critical generators to answer \mibgen{} from applying in general.
Identifying classes of closure systems, or convex geometries, where the critical base is valid is thus an interesting question to investigate in order to study the tractability of \mibgen{}.

A possible alternative is to use the lattice theoretic counterpart of critical generators, being $E$-generators~\cite{freese1995free, adaricheva2013ordered}.
Still, the corresponding $E$-base, much as the critical base, does not always constitute a valid IB of its closure system.
However, in the class of lower bounded closure systems, the $E$-base turns out to be valid.
Together with the fact that lower bounded closure systems are $D$-acyclic, this suggests our approach to \mibgen{} based on the problem \critbase{}, adapted to $E$-generators, might succeed.

As for parameters, an analogue of the notion of degeneracy coming from (hyper)graph theory would be a natural candidate for the generalization of the degree.
Degeneracy in implicational bases can be defined as the smallest left-degree among all possible left-to-right ordering of the vertices.
Dimension, also known as arity in the language of abstract convexity~\cite{van1993theory} is another intriguing parameter incomparable to degree.
It is the least integer $k$ such that the closure system admits an IB with premises of size at most $k$, which makes it particularly interesting for the translation task as \misenum{} is tractable for hypergraphs of bounded dimension.
Other parameters coming from convexity such as the Caratheodory number---for which some tractable algorithms can be obtained~\cite{wild2017joy}---, the Radon number, or the Helly number might also be avenues for future research.

Finally, we note that our algorithms for \icsenum{} have an \XP{} dependence on the degree $k$, i.e., it runs in $N^{f(k)}$ time for some computable function $f$ where $N$ is the size of the input plus the output.
This is not due to the top-down approach, that is, Theorem~\ref{thm:acsa-to-ics}, which preserves \FPT{}, i.e., running times of the form $f(k)\cdot N^{O(1)}$ for some computable function.
This is neither due to the computation of minimal transversals in Theorem~\ref{thm:acs-a-input-poly}, as this can be achieved with \FPT{} delay parameterized by the number of edges~\cite{elbassioni2008fpt}.
Rather, the bottleneck to \FPT{} time is the guess of irreducible selections from Lemma \ref{lem:conclusion-intersection-characterization}.
The same arises in Lemma~\ref{lem:intersection-progress} for the premise-degree.
Consequently, and since our algorithm for \critbase{} relies on the ones we devised for \icsenum{}, they also have an \XP{} dependence on the parameter.
Thus, it is natural to ask whether any of these algorithms can be improved to run in \FPT{} times parameterized by the degree.


\bibliographystyle{alpha}
\bibliography{arXiv} 

\end{document}